\if0
\documentclass[sigconf,prologue,table]{acmart}
\usepackage[false]{anonymous-acm}
\else
\documentclass[sigconf,prologue,table,nonacm]{acmart}
\usepackage[false]{anonymous-acm}
\fi
\AtBeginDocument{%
  }

\copyrightyear{2024} 
\acmYear{2024} 

\usepackage{enumitem}
\usepackage{amsthm}
\usepackage{physics}
\usepackage{amsfonts, fontawesome}
\usepackage{graphicx}
\usepackage[font={footnotesize}]{caption}
\usepackage[font={footnotesize}]{subcaption}
\usepackage{xcolor}
\definecolor{darkgrey}{RGB}{70,70,70}
\definecolor{lightgrey}{RGB}{200,200,200}
\definecolor{lyellow}{RGB}{255,255,100}
\definecolor{llyellow}{RGB}{250,250,180}
\definecolor{lgreen}{RGB}{144,238,144}

\usepackage{colortbl}
\newtheorem{definition}{Definition}

\usepackage[createShortEnv]{proof-at-the-end}

\usepackage{float}
\usepackage{color-edits}
\usepackage{algorithm}
\usepackage[noend]{algpseudocode}
\usepackage{multirow}
\usepackage{array}
\newcolumntype{P}[1]{>{\centering\arraybackslash}p{#1}}
\newcolumntype{M}[1]{>{\centering\arraybackslash}m{#1}}
\usepackage{bm}
\usepackage{graphicx}
\usepackage{graphbox}

\usepackage{enumitem}

\usepackage{makecell}

\usepackage{color-edits}
\addauthor[Kartik]{kl}{cyan}
\addauthor[Fabrizio]{fp}{magenta}
\addauthor[Laura]{lm}{red}
\addauthor[Kelly]{ki}{orange}
\addauthor[Maciej]{mb}{brown}
\addauthor[Torsten]{mb}{purple}

\newcommand{\fly}{PolarFly }

\newcommand{\newfly}{PolarStar }
\newcommand{\newflyN}{PolarStar}
\newcommand{\newSuper}{Inductive-Quad }
\newcommand{\newSuperN}{Inductive-Quad}
\newcommand{\faY}[0]{\faBatteryFull}
\newcommand{\faH}[0]{\faBatteryHalf}
\newcommand{\faN}[0]{\faTimes}

\newtheorem{theorem}{Theorem}
\newtheorem{corollary}{Corollary}
\newtheorem{lemma}{Lemma}
\newtheorem{property}{Property}

\newtheorem{proposition}{Proposition}
\newtheorem{propertyalt}{Property}[property]

\newenvironment{propertyp}[1]{
	\renewcommand\thepropertyalt{#1}
	\propertyalt
}{\endpropertyalt}

\begin{document}
\pagenumbering{gobble}

\title{PolarStar: Expanding the Horizon of Diameter-3 Networks}

\thanks{
This research is, in part, based upon work supported by the Office of the Director of National Intelligence (ODNI), Intelligence Advanced Research Projects Activity (IARPA), through the Advanced Graphical Intelligence Logical Computing Environment (AGILE) research program, under Army Research Office (ARO) contract number W911NF22C0081. The views and conclusions contained herein are those of the authors and should not be interpreted as necessarily representing the official policies or endorsements, either expressed or implied, of the ODNI, IARPA, ARO, or the U.S.  Government.
This work was also supported by the U.S. Department of Energy through the Los Alamos National Laboratory, operated by Triad National Security, LLC, for the National Nuclear Security Administration of U.S. Department of Energy (Contract No. 89233218CNA000001), and by the U.S. DOE LDRD program at Los Alamos National Laboratory under project number 20230692ER. The U.S. Government retains an irrevocable, nonexclusive, royalty-free license to publish, translate, reproduce, use, or dispose of the published form of the work and to authorize others to do the same. 
This paper has been assigned the LANL identification number LA-UR-22-30347, ver. 2.
}
\authoranon{
\author{Kartik Lakhotia}
\email{kartik.lakhotia@intel.com}
\affiliation{\institution{Intel}
\city{Santa Clara}
\state{CA}
\country{USA}
}

\author{
Laura Monroe}
\email{lmonroe@lanl.gov}
\affiliation{\institution{Los Alamos National Laboratory}
\city{Los Alamos}
  \state{NM}
\country{USA}
}

\author{
Kelly Isham}
\email{kisham@colgate.edu}
\affiliation{\institution{Colgate University}
\city{Hamilton}
\state{NY}
\country{USA}
}

\author{
Maciej Besta}
\email{mbesta@inf.ethz.ch}
\affiliation{\institution{ETH Z\"{u}rich}
\city{Zurich}\country{Switzerland}
}

\author{
Nils Blach}
\email{nils.blach@inf.ethz.ch}
\affiliation{\institution{ETH Z\"{u}rich}
\city{Zurich}\country{Switzerland}
}

\author{
Torsten Hoefler}
\email{torsten.hoefler@inf.ethz.ch}
\affiliation{\institution{ETH Z\"{u}rich}
\city{Zurich}\country{Switzerland}
}

\author{Fabrizio Petrini}
\email{fabrizio.petrini@intel.com}
\affiliation{\institution{Intel}
\city{Santa Clara}
\state{CA}
\country{USA}
}

\renewcommand{\shortauthors}{Kartik Lakhotia et al.}
}
\begin{abstract}
We present PolarStar, a novel family of diameter-3 network topologies derived from the star product of low-diameter factor graphs. 

PolarStar gives the largest known diameter-3 network topologies for almost all radixes, thus providing the best known scalable diameter-$3$ network.
Compared to current state-of-the-art diameter-$3$ networks, PolarStar achieves $1.3\times$ geometric mean increase in scale over Bundlefly, $1.9\times$ over Dragonfly, and $6.7\times$ over {3-D} HyperX.  
PolarStar has many other desirable properties, including a modular layout, large bisection, high resilience to link failures and a large number of feasible configurations for every radix. 

We give a detailed evaluation with simulations of synthetic and real-world traffic patterns and show that PolarStar exhibits comparable or better performance than current diameter-3 networks.
\end{abstract}




\maketitle
\section{Introduction}
\label{sec:intro}
\subsection{Motivation
}
The largest current supercomputers contain tens to hundreds of thousands of processing nodes. For example,
Frontier, the most powerful system in the Top $500$ list 
\cite{top500}, has $9,408$ CPUs and $37,632$ GPUs~\cite{frontier}. 
The fourth-ranked supercomputer, Fugaku, has $158,976$ processing nodes~\cite{Fugaku:Dongarra}.
Future supercomputers will be even larger, since the size of the system determines peak compute performance. 

Equally important is the diameter, as this {affects communication latency and injection bandwidth per \looseness=-1switch}.
Low-diameter networks, especially of diameters $2$~\cite{diameter-2-topos} and $3$~\cite{dally08}, are of great interest, providing low-latency and cost-effective high-bandwidth communication infrastructure. 	

We address here this question: what is the largest diameter-3 network topology that can be built using switches of a given radix? 

Current technological advances make this question especially timely.
The emergence of high-radix optical IO modules with high shoreline density has increased interest in scalable low-diameter networks~\cite{3d-multichip:Bergman, wade2020teraphy, darpa-eri-2019, lightmatter}.
Co-packaging of these with compute nodes on the same chip greatly enhances available bandwidth per node. Each packet consumes bandwidth on only a few links, limiting negative effects of tail latency and improving system performance~\cite{Dean:tail-latency}. 

Low-diameter networks are then required to efficiently utilize bandwidth on co-packaged chips.
Since each router is integrated with a compute 
node, scalability of a co-packaged system 
is identified
with the order of the graph defining the network \looseness=-1topology.
	
\subsection{Related Work}
{Given the importance of low-diameter topologies and the technological constraints on switch radix, 
it makes sense to design networks with the largest order~(number of nodes) possible given diameter $D$ and router degree $d$. This is the \emph{degree-diameter problem}~\cite{comb_wiki_degdiam_general}, and}
the order of $G$ is bounded above by the Moore bound~\cite{hoffmansingleton1960}. Networks with good \emph{Moore-bound 
efficiency}~(proximity to Moore bound) are not only highly scalable, but also cost-effective and power-efficient as they can 
realize a system of given size
with relatively lower radix switches and fewer cables.
Unfortunately, 
the largest known 
graphs 
for $D>2$ and $d>2$
are much smaller in size than the Moore bound.

\begin{figure}[htbp]
    \centering
    \includegraphics[width= 1\linewidth]{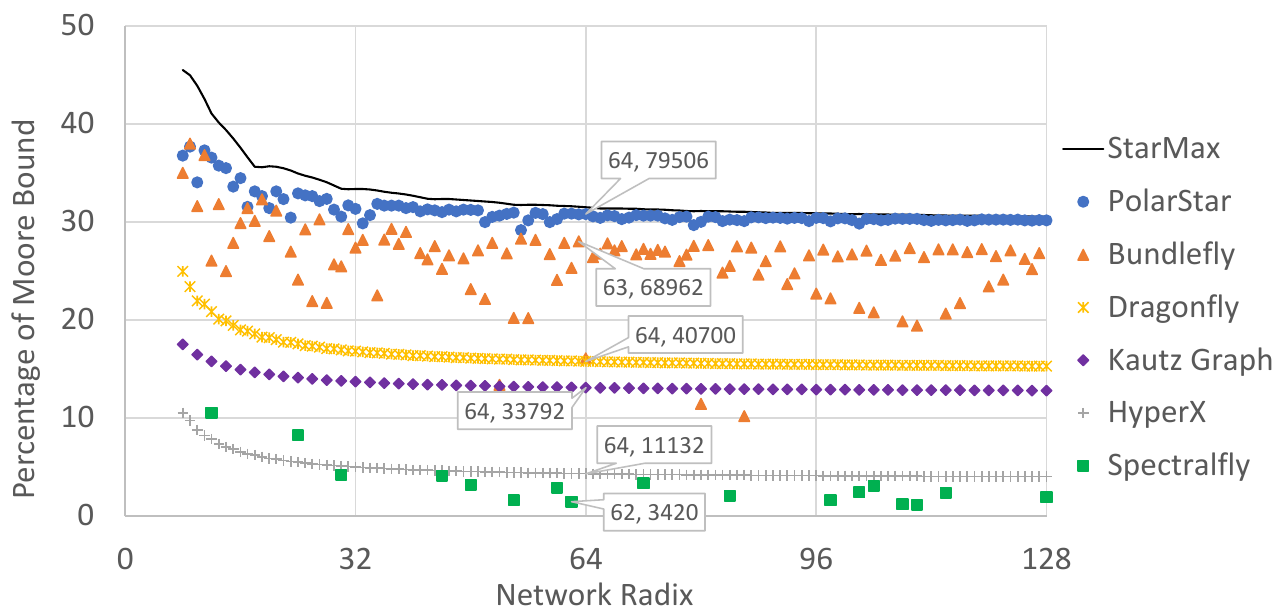}
    \caption{Scalability of direct diameter-3 topologies with respect to the Moore bound. Data labels show the largest number of nodes and corresponding radix in each topology for radix $\leq 64$. 
    \emph{StarMax} denotes an upper bound on the largest graphs theoretically achievable with $P-$ and $R-$star products -- the 
    mathematical constructs used in state-of-the-art Bundlefly and \newfly networks.
    {For Spectralfly, which is not a fixed diameter topology, we only compare  design points with diameter $\leq 3$ and largest scale for a given radix~(if it exists). For Kautz networks, we consider each link as bidirectional.}
    }
    \label{fig:MB_comparison}
    \Description{Moore-bound comparison between different networks}
\end{figure}
Some diameter-2 networks such as Polarfly~\cite{polarfly_sc22} 
and Slimfly~\cite{besta2014slim} approach the Moore bound. However, 
their scale is
limited by the small Moore-bound for diameter-2~($d^2+1$ for radix $d$).  
Since diameter-$2$ networks only span a few thousand nodes for feasible radixes, they are not suitable for massive-scale datacenters and HPC systems.

On the other hand, diameter-3 networks have a high enough Moore bound to address scalability requirements ($d^3 -d^2 +d +1$ for radix $d$). Dragonfly~\cite{dally08} and HyperX~\cite{ahn2009hyperx} are popular diameter-3 topologies deployed in large systems~\cite{lakhotia2021accelerating, frontier}, but these topologies exhibit poor Moore-bound efficiency, as shown in Figure~\ref{fig:MB_comparison},
which drives up the network cost and power consumption. 
Kautz networks~\cite{li2004graph} used in SiCortex's fabric~\cite{stewart2006new}, are directed topologies with near Moore-bound scalability for directed graphs. However, almost all deployed off-chip networking systems~\cite{foley2017ultra, pfister2001introduction, birrittella2015intel, kachris2012survey, torudbakken201350tbps} and interconnects~\cite{dally08, kim2007flattened, ahn2009hyperx, leiserson1985fat} feature bidirectional links. Treating each link of a Kautz network as bidirectional doubles the degree of nodes,  resulting in a higher Moore-bound and lower efficiency. For diameter-3, bidirectional Kautz networks have $<13\%$ asymptotic Moore-bound efficiency, which is $2.4\times$ smaller than PolarStar.

Recently, Lei et al. introduced Bundlefly~\cite{lei2020bundlefly}, a 
modular diameter-3 network based on star product of two graphs. Its 
inter-module links can be bundled into multi-core fibers to reduce cabling 
cost. However, its Moore-bound efficiency
varies, as shown in Figure~\ref{fig:MB_comparison}.
    
Creating largest diameter-$3$ graphs is an open problem in mathematics.
There is a big gap between the best-known diameter-$3$ graphs and the 
Moore bound. The Combinatorics Wiki leaderboard~\cite{comb_wiki_degdiam_general} lists some of the largest diameter-3 graphs as of 2013, up to radix $20$. 
The construction in \cite{bermond82}, also utilized by state-of-the-art Bundlefly~\cite{lei2020bundlefly}, surpasses the leaderboard for radixes $18-20$. To the best of our knowledge, the construction in \cite{bermond82} also gives the largest known graphs for most radixes $>20$, before our construction. 

In this paper, we construct a new family of diameter-3 graphs that improve upon the previous best~\cite{bermond82, comb_wiki_degdiam_general} for almost all radixes $\geq18$, achieving a $1.3\times$ geometric mean higher scale than Bundlefly~\cite{lei2020bundlefly}. 
 \subsection{Contributions 
}
We propose a new family of network topologies called \newfly that extends PolarFly~\cite{polarfly_sc22} to large diameter-3 networks using a mathematical construct called the star product.
\begin{itemize}[itemsep=0pt,parsep=2pt,leftmargin=*]     
    \item \newfly gives the \emph{largest known diameter-$3$ direct networks}
    for almost all radixes, achieving $1.3\times, 1.9\times$ and $6.7\times$ geometric mean increase in scale over Bundlefly, Dragonfly and HyperX. The base graphs are some of the \emph{largest known diameter-3 graphs.}
    
    \item \newfly reaches \emph{near-optimal scalability} for diameter-3
    star products that can be constructed with the currently known properties of the factor graphs that give low-diameter star-products. 
    Further optimizations on these diameter-limited star products are unlikely to provide notable benefits. 

    \item \newfly \emph{extends several networking benefits of PolarFly}, including a modular layout amenable to bundling of links into multi-core fibers and a 
    large bisection cut. 
    
    \item \newfly has a \emph{large design space} and it exists with
    multiple configurations
    for every radix in $[8,128]$.  
    \item We design a \emph{routing mechanism} for \newfly using its theoretical properties, and show \emph{comparable or better performance than existing diameter-3 networks on synthetic and real-world motifs}.

\end{itemize}

\section{Background}
\label{sec:back}
\subsection{Network Model}
In direct networks, each switch is directly linked to or integrated with a compute endpoint.
Hence, the topology of a direct network can be modeled as an undirected graph $G(V,E)$ where $V(G)$ is
the set of switching nodes, or vertices, 
$\abs{V(G)}$ 
is the order of $G$, and $E$ is the set of links, or edges.
Each node has $d$ links to other
nodes where $d$ is the  \emph{network radix}, or \emph{degree}.
The maximum length of shortest paths between
any node pair is the
\emph{diameter}~$D$. 
In this paper, we consider networks of diameter $3$.
\subsection{The Moore Bound}\label{sec:moore_bound}
Moore bound \cite{hoffmansingleton1960} is an upper bound on the number of nodes $N$ that a graph of degree $d$ and diameter $D$ may have. This bound is
\begin{equation*}\label{moore_bd}
N \le 1+d\cdot\sum_{i=0}^{D-1}(d-1)^i.
\end{equation*}
For diameter-$3$, the Moore bound is
$
N \le d^3 - d^2 + d + 1.
$

The only graphs with $D\geq 2$ and $d\geq 2$ that achieve the Moore bound are cycles, the Hoffman-Singleton graph, the Petersen graph \cite{hoffmansingleton1960,Bannai1973OnFM,Damerell1973OnMG} and a hypothetical diameter-$2$ degree-$57$ graph \cite{Dalfo20191full}.
These graphs are not suited for large-scale network design: degree-2 cycles with low diameter are very small, and the others have only one design point each.
Few graphs even come close to the Moore bound. The latest leaderboard of degree-diameter problem from~\cite{comb_wiki_degdiam_general} shows that 
the best known diameter-3 graphs, 
with the most relevant degrees, reach only $25-30\%$ of the Moore bound.
The \newfly construction proposed here is larger than the best known graphs, which are discussed in \cite{bermond82} and \cite{lei2020bundlefly}.
\subsection{Network Properties}\label{sec:desirable_properties}
In this section, we compare PolarStar with  Slimfly~\cite{besta2014slim}, PolarFly~\cite{polarfly_sc22}, Dragonfly~\cite{dally08}, HyperX~\cite{ahn2009hyperx}, MegaFly~\cite{flajslik2018megafly, shpiner2017dragonfly+}, SpectralFly~\cite{aksoy2021spectralfly}, Bundlefly~\cite{lei2020bundlefly} and
Fat-trees, on 
the basis of several desirable network properties. 
This comparison may be seen in Table~\ref{tab:properties}.

\begin{table}[ht] 
\setlength{\tabcolsep}{3.5pt}
\centering
\footnotesize
\renewcommand{\arraystretch}{0.6}
\resizebox{\linewidth}{!}{%
\begin{tabular}{lccccc@{}}
\toprule
\textbf{Topology} & \makecell[c]{\textbf{Direct}} & \textbf{Scalability} & 
\makecell{\textbf{Stable }\textbf{Design-space}} 
& \textbf{$D\leq 3$} & \textbf{Bundlability}  \\
\midrule
Fat-tree & \faN & \faY & \faY & \faN & \faY\\

PolarFly & \faY & \faN & \faH & \faY & \faY\\

Slimfly & \faY & \faN & \faH & \faY & \faY \\

3-D HyperX & \faY & \faH & \faY & \faY  & \faY \\

Dragonfly & \faY & \faY & \faY & \faY & \faH \\

Bundlefly & \faY & \faY & \faH & \faY & \faY\\

{Megafly} & \faN & \faY & \faY & \faY & \faH\\

{Spectralfly} & \faY & \faH & \faH & \faY & \faH\\
\midrule

\textbf{\newfly} & \faY & \faY  & \faY & \faY & \faY\\
\bottomrule
\end{tabular}
}
\caption{
Network Properties assessment: battery levels represent network's standing in terms of the corresponding property ``\faY'': very good, ``\faH'': fair, ``\faN'': not good.
$D$ is Network Diameter}
\label{tab:properties}
\end{table} 
\noindent\textbf{\underline{Directness}:}
Every switch in a direct network is attached to one or more endpoints, 
whereas indirect networks have some switching nodes not attached to any endpoint.
Directness is a desirable property; for example --
{if co-packaged modules are used -- \emph{indirect} networks such as Fat-tree and MegaFly require fabricating two types of chips, which increases their cost. 
Further, a switch-only chip in these topologies theoretically requires twice the number of ports than a co-packaged chip with an endpoint.}

\noindent\textbf{\underline{Scalability}:} Absolute scale achieved by a topology depends upon the Moore bound for its diameter and its Moore-bound efficiency. 
Diameter-2 networks 
such as PolarFly~\cite{polarfly_sc22}
approach the Moore bound but are limited in 
scale as the bound is small for diameter $2$. Three-level Fat-trees scale similarly to diameter-2 networks 
but additional levels can be used to increase scalability.

Diameter-3 networks have a large Moore bound that can cover hundreds of thousands of nodes with common switch radixes. 
However, known diameter-3 networks such as HyperX and SpectralFly have poor Moore-bound efficiency. PolarStar has the highest Moore-bound efficiency among diameter-3 network topologies.

\noindent\textbf{\underline{Low-diameter}:}
Networks with small diameter enable low-latency remote accesses
and can sustain high ingestion bandwidth per switch. 
Diameter $\leq 3$ is
preferred as it provides both performance and scalability at low cost.

\noindent\textbf{\underline{Stable Design-Space}:} A desirable
topology provides several configurations for all radixes,
and provides stable scaling (consistent Moore-bound efficiency without jitters). This allows system design at various scales (dictated by customer demands) with limited choices for switch radixes (constrained by redesign costs, availability etc.). Slimfly and PolarFly~\cite{besta2014slim}
have few feasible radixes and configurations.
Bundlefly's~\cite{lei2020bundlefly} Moore-bound efficiency fluctuates significantly with the radix (Figure~\ref{fig:MB_comparison}). PolarStar features several configurations for every radix and has a more stable Moore-bound efficiency.
{Although SpectralFly can be constructed for many radixes, it has diameter-$3$ for very few radixes, as shown in Figure~\ref{fig:MB_comparison}.}

\noindent\textbf{\underline{Bundlability}:} 
Bundlefly~\cite{lei2020bundlefly} introduced the concept of bundling friendly networks as modular topologies with  
multiple links between adjacent modules~(logical groups of nodes) that can be packed together in a Multi-Core Fiber~\cite{awaji2013optical}. 
This reduces cabling cost and complexity. Both Bundlefly and PolarStar are amenable to bundling as per this description.
{The largest Dragonfly and Megafly constructions are less bundling friendly -- there is one link between each pair of node groups. Multiple links can be used between adjacent groups to make these topologies bundlable~\cite{frontier}. However,} {doing so} {reduces their scalability and Moore-bound efficiency proportionally.} 
Additionally, large racks that can accommodate multiple logical groups can also enable bundling inter-rack cables.
\section{Approach}
The \emph{star product} is a non-commutative product 
of two factor graphs, presented by Bermond, Delorme and Farhi in \cite{bermond82}. 
Under certain conditions, the product can achieve a lower diameter than the sum of the diameters of the factors. Lei et al. used a star-product construction for Bundlefly \cite{lei2020bundlefly}. 

In this paper, we also use the star product, but devise new properties for our factor graphs that differ from the original ones from \cite{bermond82} used for Bundlefly. These new properties permit larger factor graphs of our own design, and allow us to construct the \newfly star-product topology, which is larger than Bundlefly and all other previously known diameter-3 networks. 

The factor graphs we use not only have large order but also have many other qualities that carry over to the star product, which are useful for networking. We analyze these and the networking characteristics of PolarStar throughout the paper.

\section{Star Products}
\subsection{Intuition for the Star Product}
The star product generalizes the Cartesian product $G \times H$.
In the Cartesian product, copies of $H$, called \emph{supernodes}, replace vertices of \emph{structure graph} $G$. Edges join corresponding vertices in any two copies $H_1$ and $H_2$, when $H_1$ and $H_2$ are joined by an edge in $G$. Figure~\ref{fig:const_cart} illustrates this.

The edges of a star product are also determined by bijections between vertices in neighboring supernodes. In this case, however, the requirement that edges be constructed between corresponding vertices is relaxed. The bijections need not even be the same between different pairs of neighboring supernodes. Carefully selecting the bijective connections may provide a product of diameter at most 1 more than that of the structure graph.
We will later discuss how the selection of bijections permits this constrained diameter.

Intuitively, the star product retains the large-scale form of the structure graph $G$
and embeds copies of the supernode $H$ in place of vertices of $G$. The supernodes are then connected bijectively, with edges as desired. 
Figure~\ref{fig:star_example} compares a Cartesian product with an example star product on the same factor graphs with its large-scale structure graph and supernode embeddings and connectivity. Figure~\ref{fig:star_ps} shows our star product construction.  
\begin{figure}[!htbp]
    \centering
    \begin{subfigure}{1\linewidth}
      \centering
      \includegraphics[width=.8\linewidth]{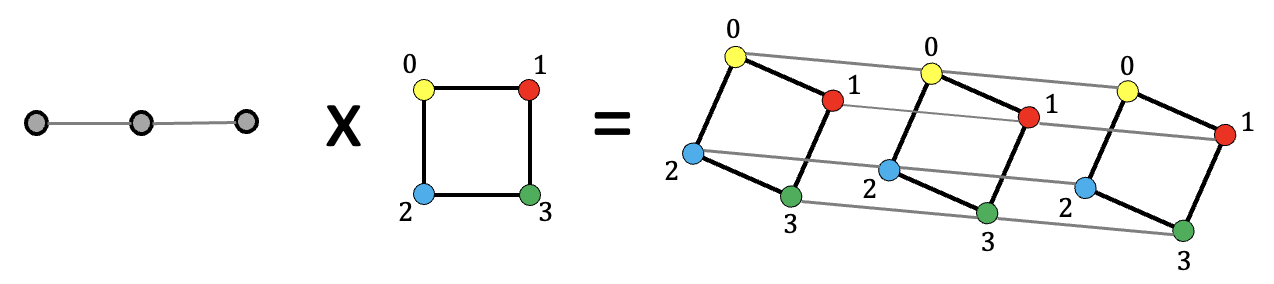}
      \caption{The Cartesian product $L_3 \times C_4$. Corresponding nodes in neighboring copies of $C_4$ are joined by an edge. The bijection is simply the identity.}
      \label{fig:const_cart}
    \end{subfigure}
    \begin{subfigure}{1\linewidth}
      \centering
      \includegraphics[width=.8\linewidth]{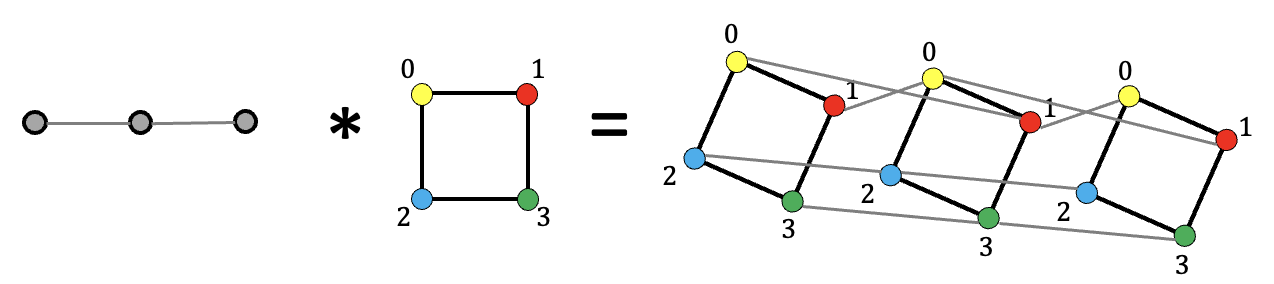}
      \caption{An example star product $L_3 * C_4$.
      In this particular star product, neighboring supernodes are all joined with the same bijection $f=(0 1)(2)(3)$.
      }
      \label{fig:const_star}
    \end{subfigure}
\caption{A comparison of the Cartesian product $L_3 \times C_4$ with an example star product $L_3 * C_4$. The structure graph $L_3$ is the path graph on three vertices, and the supernode $C_4$ is the cycle graph on four vertices. }
\label{fig:star_example}
\Description{Comparison of Cartesian product with star product}
\end{figure}

\begin{figure*}[t]
\centering
\begin{subfigure}[t]{0.22 \textwidth}
\centering
\includegraphics[height=.7\columnwidth]{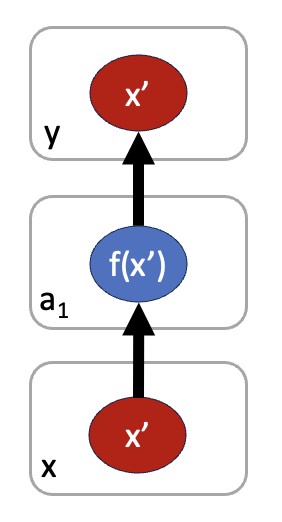}
%
\caption{$y'=x'$: An $x'$-alternating path of length $2$ through the structure graph.}
\label{fig:diameter_setup}
\end{subfigure}
\quad
\begin{subfigure}[t]{0.22 \textwidth}
\centering
\includegraphics[height=.95\columnwidth]{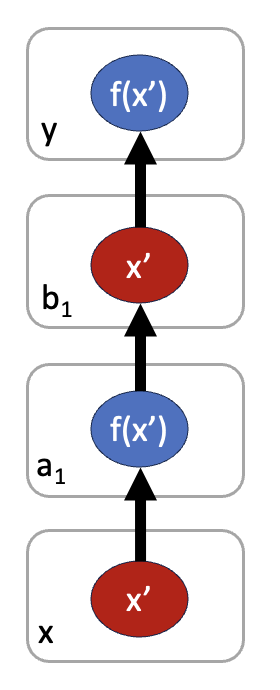}
\caption{$y'=f(x')$: An $x'$-alternating path of length $3$ through the structure graph, using the length-$2$ $x'$-alternating path between $a_1$ and $y$.}
\label{fig:diameter_p1}
\end{subfigure}
\quad
\begin{subfigure}[t]{0.22 \textwidth}
\centering
\includegraphics[height=.7\columnwidth]{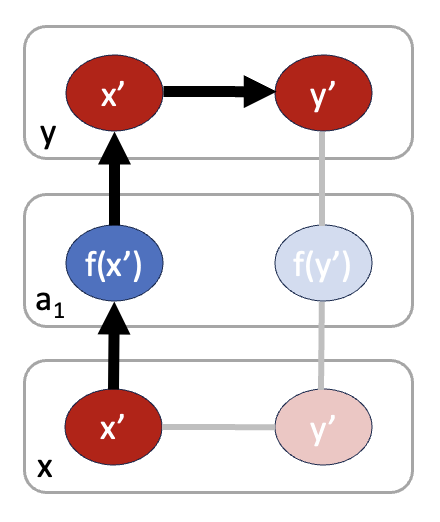}
\caption{$(x',y') \in E(G')$: $x'$- and $y'$-alternating paths of length $2$ through the structure graph, with a hop through supernode $y$ (or $x$).}
\label{fig:diameter_p2}
\end{subfigure}
\quad
\begin{subfigure}[t]{0.22 \textwidth}
\centering
\includegraphics[height=.7\columnwidth]{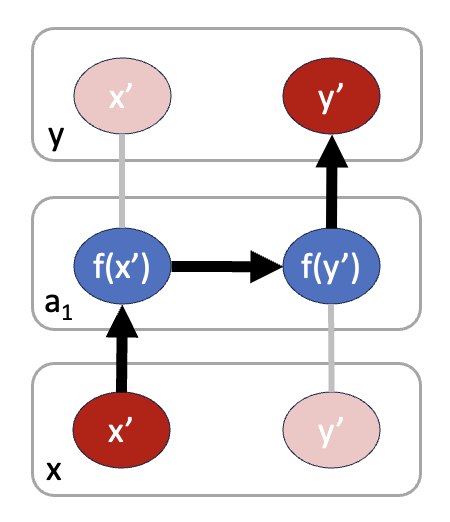}
\caption{$(f(x'),f(y')) \in E(G')$: $x'$- and $y'$-alternating paths of length $2$ through the structure graph, with a hop through supernode $a_1$.}
\label{fig:diameter_p3}
\end{subfigure}
\quad
\caption{This figure illustrates Theorem~\ref{thm_diam-3} on a diameter-2 structure graph with property $R$, showing how paths are constructed from $2$- and $3$-hop alternating paths. The star product path may include a $1$-hop supernode detour, shown in \ref{fig:diameter_p2} and \ref{fig:diameter_p3}, which is permitted by Property~\ref{prop_R_star} on $G'$. 
}
\label{fig:diameter_theorem}
\Description{A figure showing construction of diameter-(D+1) paths}
\end{figure*}

\subsection{Formal Definition}
The following formal Definition~\ref{def:star_product} is taken from \cite{bermond82}. In this paper, we call $G$ the \emph{structure graph} and $G'$ the \emph{supernode}.
\begin{definition}\label{def:star_product}~\cite{bermond82} 
    Let $G=(X,E)$ and $G'=(X',E')$ be two graphs . Take an arbitrary orientation of the edges of $G$ and let $U$ be the resulting set of arcs. For each arc $(x,y)$ in $U$, let $f_{(x,y)}: X'\rightarrow X'$ be a bijection. 
The star product $G_* = G*G'$ is defined as follows:
\begin{enumerate}[itemsep=0pt,parsep=0pt]
    \item The vertex set of $G_*$ is the Cartesian product $X\times X'$.
    \item Vertex $(x,x')$ is joined to vertex $(y,y')$ if and only if either
    \begin{enumerate}[itemsep=0pt,parsep=0pt]
        \item $x=y$ and $(x',y')\in E'$, or \item $(x,y)\in U$ and $y'=f_{(x,y)}(x')$.
    \end{enumerate}
\end{enumerate}
\end{definition}

In this definition, condition $2a$ replaces nodes of $G$ with copies of $G'$. Condition $2b$ joins copies of $G'$ to each other. The bijection $f_{(x,y)}$ in $2b$ is the rule for bijective connectivity from vertices of $G$ to vertices of $G'$ across the arc $(x,y)$.

\subsection{Order, Degree, and Diameter}\label{sec:genprops_star}
Using the following facts, one may construct a large diameter-constrained graph from two smaller graphs.
\begin{enumerate}[itemsep=0pt,parsep=2pt]
    \item The number of vertices is $\abs{V(G_*)}=\abs{V(G)}\abs{V(G')}$.

    \item If the maximum degrees in $G$ and $G'$ are $d$ and $d'$ respectively, the maximum degree of $G_*$ is $d_* \leq d + d'$.
    
    \item If the diameters of $G$ and $G'$ are $D$ and $D'$, respectively, then the diameter of $G_*$ is given by $D_* \leq D + D'$.
\end{enumerate}
Appropriate bijections for connectivity between adjacent supernodes can reduce the diameter $D_*$ to at most $D+1$. In this paper,
we will use the same bijection $f$ for all arcs of the structure graph, and construct a star product with this restricted diameter. In Section \ref{sec:construction_R_star_products}, we describe such a bijection.

\section{A Low-Diameter Star Product}
\label{sec:new_properties} Section~\ref{sec:genprops_star} 
gives upper bounds on degrees and diameters of star products. In Section~\ref{sec:properties_def}, we define properties on the factor graph that allow us to construct a star product with diameter at most $1$ more than that of the structure graph.

Bermond et al. describe properties P and P$^*$ on graphs $G$ and $G'$ giving large star products with minimal or no increase in diameter over that of $G$ \cite{bermond82}.
Our new properties \ref{prop_R} and \ref{prop_R_star} are similar in spirit 
but are not the same.

In particular, our \ref{prop_R_star} weakens P$^*$ by allowing any diameter, unlike P$^*$ which restricts diameter of $G'$ to $\le 2$. Our \ref{prop_R} strengthens P by requiring that all vertex pairs be joined by a path of length diameter $D$ and consequently, by a path of length $D+1$. In contrast, P only requires that vertices with a path of length $D$ between them be joined by a path of length $D+1$ as well.

Weakening P$^*$ permits us to design $IQ$, a $G'$ supernode graph with Property~\ref{prop_R_star} that can have diameter $>2$ and a larger order than previously known supernodes with Property P$^*$. Using $IQ$, and noting that Erd\H os-R\'enyi polarity graphs have Property~\ref{prop_R},
we produce a 
\emph{larger star product} having diameter 3. 

Using the new $R$ properties, 
we thus improve on the two largest known star-product graphs given in \cite{bermond82}, with $1.06\times$ and $1.21\times$ geometric mean improvement for radixes in the range $[8,128]$. 
In fact, this new family of diameter-$3$ graphs is
\emph{larger than any diameter-3 graphs previously designed} for almost all radixes.

More importantly for this community, PolarStar is much larger than state-of-the-art Bundlefly \cite{lei2020bundlefly}, which is a $P$-star-product based on \cite{bermond82}. Bundlefly does not use our 
larger supernode, nor does it use an optimal structure graph. \newfly thus achieves $1.3\times$ geometric mean larger scale than Bundlefly, for network radixes in $[8,128]$.
\subsection{Useful Factor-Graph Properties}\label{sec:properties_def}
\subsubsection[Structure-Graph Property R]{Structure-Graph Property $R$}
This first property and its associated proposition apply to the structure graph $G$, and highlight its \emph{path diversity}. 
We will later see that this diversity enables reachability in $D+1$ hops between all vertex pairs in the star product.
\begin{propertyp}{R}\label{prop_R}
A graph $G$ of diameter $D$ has Property~\ref{prop_R} if any vertex pair $x,y\in V(G)$ 
can be joined by a path of length~$D$.  
\end{propertyp}
Note that in the definition of Property~\ref{prop_R}, self-loops (if they exist in $G$) are permissible as part of the length-$D$ path.
\subsubsection[Supernode Properties R* and R1]{Supernode Properties $R^*$ and $R_1$}
The next properties, \ref{prop_R_star} and \ref{prop_R_1}, apply to the supernode $G'$. 
These properties address \emph{paths internal to  supernodes}. 
Given that paths between supernodes are defined by application of a bijection $f$ to supernode vertices, we may use these to define a short path through the star product.

\begin{propertyp}{R$^*$}\label{prop_R_star}Graph $G'$ 
satisfies Property~\ref{prop_R_star} if there is an involution $f$ so that for any $x',y' \in V(G')$, at least one of the following is \looseness=-1true: 
\begin{enumerate}[label=(\alph*),itemsep=1pt,parsep=2pt]    
    \item $y'=x'$
    \item $y'=f(x')$\label{prop_R_star_eq}
    \item $(x',y') \in E(G')$\label{prop_R_star_inE}
    \item $(f(x'),f(y')) \in E(G')$\label{prop_R_star_inEprime}
\end{enumerate}
\end{propertyp}

\begin{proposition}\label{proposition:Rstar}
    A graph $G'$ of degree $d'$ having Property~\ref{prop_R_star} has at most $2d'+2$ vertices.
\end{proposition}
\begin{proof}
   Fix some vertex $y$ in $G'$. We will count how many other vertices there can be. Let $x$ be a vertex in $G'$. By Property~\ref{prop_R_star}, at least one of the following is true:
	\begin{enumerate}[itemsep=1pt,parsep=2pt]
	\item $y = x$, giving exactly one choice for $x$.
	\item $y= f(x)$: Since $f$ is a bijection, then there is exactly one choice for $x$, namely $f^{-1}(y)$.
	\item $(x,y) \in E(G')$, giving $\deg(y)$ choices for $x$.
	\item $(f(x), f(y)) \in E(G')$: since $f$ is a function, $f(y)$ is fixed. Therefore there are $\deg(f(y))$ choices for $x$.
   \end{enumerate}
   Altogether, we have at most $2 + \deg(y) + \deg(f(y))$ total vertices, assuming all conditions are disjoint. Since $\deg(y) \le d'$ and $\deg(f(y)) \le d'$, $G'$ can have at most $2+2d'$ vertices.    
\end{proof}\begin{definition}
    Let $G*G'$ be a star product where $G'$ has Property~\ref{prop_R_star}. A \emph{$x'$-alternating path} in $G*G'$ is a path between supernodes of the form $((a_0,x'), (a_1,f(x')), (a_2,x'), \dots)$, where $(a_0, a_1, a_2, ...)$ is a path in $G$, and the second entries in the path elements alternate between $x'$ and $f(x')$.
\end{definition}
When $G'$ has Property~\ref{prop_R_star}, the alternating path is the only kind of path that exists between supernodes, by definition of the star product. This is because $f$ is an involution i.e. $f^2(x')=x'$.

Subproperties \ref{prop_R_star_inE} and \ref{prop_R_star_inEprime} of Property~\ref{prop_R_star} allow a hop from any alternating path to any other.
This is quite powerful, and is key to Theorem~\ref{thm_diam-3}, establishing the low diameter of these star products.
\begin{lemma}\label{lemma:altpaths}
    Let $G*G'$ be a star product where $G'$ has Property~\ref{prop_R_star}. For all $x' \in G'$, there is an alternating path in $G*G'$ for every path in $G$. All alternating paths are of this form.
\end{lemma}

\begin{proof}
    Follows from the definition of the bijective star product, and the fact that $f^2(x')=x'$ for all $x'$.
\end{proof}
The following Property \ref{prop_R_1} is Property $P_i$ in \cite{bermond82} for the special case $i=1$, which is needed for diameter-$3$ star products.

\begin{propertyp}{R$_1$}\cite{bermond82}\label{prop_R_1}  A graph $G'$  has Property~\ref{prop_R_1} if there is a bijection $f$, with $f^2$ an automorphism of $G'$, so that the set of edges 
$
E(G') \cup f(E(G'))
$
is the entire set of edges in the complete graph on \looseness=-1$V(G')$.
\end{propertyp}
\subsection{Star Products of Low Diameter}\label{sec:construction_R_star_products}
Here we give criteria for a star product $G*G'$ to have diameter at most $D+1$. This occurs when $G$ has diameter $D$, \emph{and either}
\begin{itemize}[itemsep=1pt,parsep=2pt]
        \item The structure graph $G$ has Property~\ref{prop_R} and the supernode $G'$ has Property~\ref{prop_R_star} (Theorem~\ref{thm_diam-3}), \emph{or}
        \item The supernode $G'$ has Property~\ref{prop_R_1} (Theorem~\ref{thm_R1},~\cite{bermond82}).
 \end{itemize}
We define a path in the star product from $(x,x')$ to $(y,y')$. The strategy here is to identify $x'$- and $y'$-alternating paths of length $D$ 
from supernodes $x$ (or neighbors of $x$) to $y$. We take a single hop from the $x'$ to the $y'$ alternating paths if necessary, using Subproperties \ref{prop_R_star_inE} and \ref{prop_R_star_inEprime} from Property~\ref{prop_R_star}.

Figure~\ref{fig:diameter_theorem} illustrates Theorem~\ref{thm_diam-3}, showing the diameter-$3$ problem of interest for this paper. 
The theorem, its proof, and the figure give insight into routing strategies discussed in Section~\ref{sec:minpath}. 

\begin{theorem}\label{thm_diam-3}
    Let $G$ and $G'$ be graphs that satisfy Property~\ref{prop_R} and \ref{prop_R_star}, respectively. Define $f_{(x,y)}(x') = f(x')$ (from Property~\ref{prop_R_star}) for every edge $(x,y)$ of $G$. If diameter of $G$ is $D$, the star product $G_* = G*G'$ is a graph with diameter at most $D+1$.
\end{theorem}
\begin{proof}
    Let $(x,x')$ and $(y,y')$ be nodes in $G*G'$. By 
Property~\ref{prop_R} 
of $G$, 
 there are length-$D$ $x'-$ and $y'-$alternating paths joining supernodes $x$ and $y$, and also joining supernodes $a$ and $y$
in $G*G'$, where $a$ is a neighbor of $x$ in $G$. 

We consider the exhaustive set of cases relating $x'$ and $y'$ as described in Property~\ref{prop_R_star}. For simplicity of presentation, we assume that $D$ is even. The same proof holds for $D$ odd, reversing the reasoning for cases~\ref{th:diam-3_a} and \ref{th:diam-3_b}. 
\begin{enumerate}[label=(\alph*),itemsep=0pt,parsep=2pt]
    \item If $y'= x'$, we stay in the length-$D$ $x'$-alternating path. \label{th:diam-3_a}
    \item If $y'= f(x')$, we take $1$ hop to $(a,f(x'))$ ($a$ a neighbor of $x$ in $G$), then stay in the length-$D$ $x'$-alternating path from $(a,f(x'))$ to $(y, f(x'))$, giving a length $D+1$ path.\label{th:diam-3_b} 
\end{enumerate}
In the remaining cases, we start with a $D$-hop $x'$- alternating path and take $1$ hop in some supernode to a parallel $D$-hop $y'$- alternating path, giving a length $D+1$ path. Supernode choice depends on which of  $(x',y')$ or $(f(x'),f(y'))$ is in $E(G')$.
\begin{enumerate}[resume, label=(\alph*),itemsep=0pt,parsep=2pt]
    \item If $(x',y')\in E(G')$, we may take $1$ hop from $(b,x')$ to $(b,y')$ at any node $(b,x')$ in the $x'$-alternating path. We reach such $(b,x')$ on even hops.
    \label{th:diam-3_c}
    \item If $(f(x'),f(y'))\in E(G')$, we may take $1$ hop from $(b,f(x'))$ to $(b,f(y'))$ at any node $(b,f(x'))$ in the $x'$-alternating path. We reach such $(b,f(x'))$  on odd hops.
    \label{th:diam-3_d}
\end{enumerate}
\end{proof}
Each of the cases in the proof of Theorem~\ref{thm_diam-3} are illustrated in Figure~\ref{fig:diameter_theorem}, for $D=2$. Figure~\ref{fig:diameter_p2} illustrates  case~\ref{th:diam-3_c}, where the intra-supernode hop is taken $2$ hops into the path at supernode $y$. We might also have taken this hop at supernode $x$ which is $0$ hops into the path. Likewise, Figure~\ref{fig:diameter_p3} illustrates  case~\ref{th:diam-3_d}, where the intra-supernode hop is taken $1$ hop into the path at supernode $a_1$. 

The following theorem from \cite{bermond82} was stated for Property $P_i$. We restate the theorem for our special case where $i=1$, to get diameter $D+1$. The proof is the same as that found in \cite{bermond82}.
\begin{theorem}\label{thm_R1}\cite{bermond82}
Let $G$ be a graph of diameter $D \ge 2$, and let $G'$ have Property~\ref{prop_R_1}. Define $f_{(x,y)}(x') = f(x')$ (from Property~\ref{prop_R_1}) for every arc $(x,y)$ of an arbitrary orientation of the edges of $G$. Then $G*G'$ has diameter at most $D+1$.
\end{theorem}
\section{Choosing Good Factor Graphs}\label{sec:structure_supernode}
We want to choose the largest feasible structure graph and supernode, since the order of the star product is the product of the orders of the factor graphs. We also want factor graphs with properties useful for networking  (including the $R$ properties, among others).
\begin{itemize}[itemsep=0pt,parsep=2pt,leftmargin=*]
    \item For the structure graph, we use the \emph{Erd\H os-R\'enyi} ($ER_q$) family of polarity graphs, introduced by Erd\H os and R\'enyi  \cite{erdosrenyi1962} and 
    by Brown \cite{brown_1966}. This family of graphs has Property~\ref{prop_R} and is larger than any other known family of diameter-$2$ graphs, 
    so is an excellent structure graph candidate. $ER_q$ has many other networking advantages that the star-product inherits~\cite{polarfly_sc22, lakhotia2023network, dawkins2024edge}.
    \item For the supernode, we construct a new graph having Property~\ref{prop_R_star} called \emph{\newSuperN} (IQ). IQ is larger than 
    all other supernodes with properties known to produce diameter-3 star product with $ER_q$, and
    thus, is a better $G'$ candidate than any existing construction. IQ attains the bound on R$^*$ graphs; thus, no supernode with property R$^*$, P$^*$ or P$_1$, can be larger than $IQ$. 
\end{itemize}
\subsection{Choice of the Structure Graph \texorpdfstring{$G$}{} }\label{sec:ER}
\subsubsection[ERq and its Networking Properties]{$ER_q$ and its Networking Properties}
The Erd\H os-R\'enyi (ER) or Brown family of polarity graphs \cite{erdosrenyi1962,brown_1966} is based on finite projective geometry, where adjacency is defined by orthogonality. These graphs were used for the \fly\cite{polarfly_sc22}, SymSig~\cite{brahme2013symsig} and demi-PN~\cite{camarero2016projective} networks, due to their many advantages~\cite{polarfly_sc22,lakhotia2023network}.

These desirable networking properties carry over to the star product. Thus, the family of ER graphs is an excellent candidate for the structure graph, since these graphs do have Property~\ref{prop_R}, as per Theorem~\ref{thm:ER_PropR}. Advantages include:
\begin{itemize}[itemsep=0pt,parsep=2pt,leftmargin=*]
    \item $ER$ graphs are larger than other known diameter-$2$ graphs at almost all degrees, as shown in Figure~\ref{fig:er_mb}. The order of $ER$ graphs for degree $q+1$ is $q^2+q+1$, so they asymptotically reach the diameter-$2$ Moore bound ($q^2+2q+2$).
    \item $ER$ has a large number of feasible degrees: an $ER$ graph exists for any degree $q+1$, where $q$ is a prime power. 
    \item $ER$ has a modular layout.
    \item $ER$ has high bisection bandwidth.
\end{itemize}
\begin{figure}[!htbp]
    \centering
    \includegraphics[width=.8\linewidth]{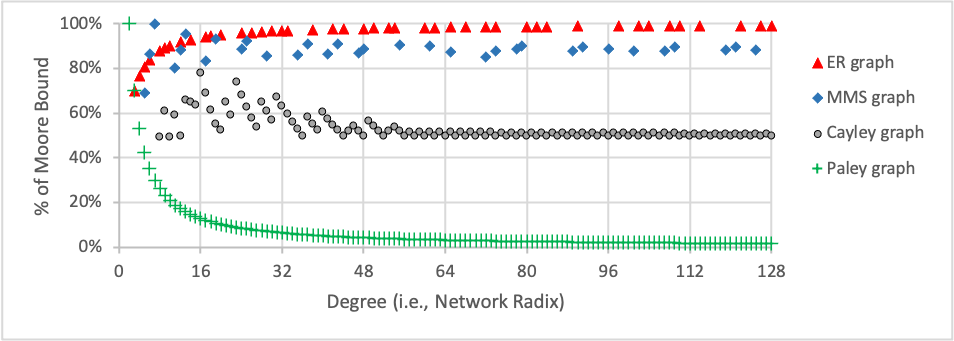}
    \caption{Moore-bound comparison for some known families of diameter-$2$ graphs: the $ER$ graph, the McKay-Miller-\v{S}ir{\'a}\v{n} graphs~\cite{mckay98}, the best Cayley graphs~\cite{abas_2017}, and the Paley graph. 
    It can be seen that any larger structure graph would only marginally increase the size of the star product.
    }
    \label{fig:er_mb}
    \Description{Moore-bound comparison}
\end{figure}
\subsubsection[Construction of ERq]{Construction of $ER_q$}
The vertices of $ER_q$ are vectors $(x,y,z)$, with $x,y,z \in \mathbb{F}_q$, the finite field of order $q$. 
Vertices $v$ and $w$ are adjacent if their dot product $v\cdot w=0$, with addition and multiplication as in $\mathbb{F}_q$. 
Since adjacency is defined
by orthogonality of two vectors, all multiples of any two vectors retain the
same adjacency relationship. Thus, we move to projective space,
considering only the left-normalized form  of each vector (so the leftmost non-zero entry of each vector is $1$).
The ER graph has these left-normalized vectors as the vertices and edges between all orthogonal vector pairs.
Note that arithmetic over $\mathbb{F}_q$ is used to compute orthogonality. See \cite{mceliece_1987}
for details of arithmetic over $\mathbb{F}_q$ and \cite{polarfly_sc22} for ER graph \looseness=-1construction.

$ER_q$ is a diameter-$2$ graph. This may intuitively be seen by considering perpendicularity in Euclidean space. 
Each pair of distinct vectors $v_0$ and $v_1$ is orthogonal to a common ${w=v_0 \times v_1}$, the cross product of $v_0$ and $v_1$. The $2$-hop path from $v_0$ to $v_1$ is then given by ($v_0,w,v_1$). The intuition is similar in the case of finite geometry.

Certain vertices in $ER_q$ are self-orthogonal, since we use the arithmetic of $\mathbb{F}_q$. If we allow the self-loops as factor-graph edges, the condition of Property~\ref{prop_R} then holds for all vertices.

\begin{theorem}\label{thm:ER_PropR}
    ER$_q$ has Property~\ref{prop_R} for all prime powers $q$.
\end{theorem}
\begin{proof}
    Any pair of vertices $(x,y)$ in an ER graph are connected by a $2-$hop path via the cross product vertex $w=x \times y$. For some pairs $\{x,y\}$, $w$ may be the same as $x$ or $y$, in which case the $2-$hop path includes a self-loop.    
\end{proof}

Factor graph self-loops then add edges to the star product, as in Figure~\ref{fig:star_paley}, and we drop any remaining self-loops in the product.

\subsection{Choice of the Supernode \texorpdfstring{$G'$}{}}
\begin{figure}[!htbp]
    \centering
    \begin{subfigure}{.5\linewidth}
      \centering
      \includegraphics[width=.5\linewidth]{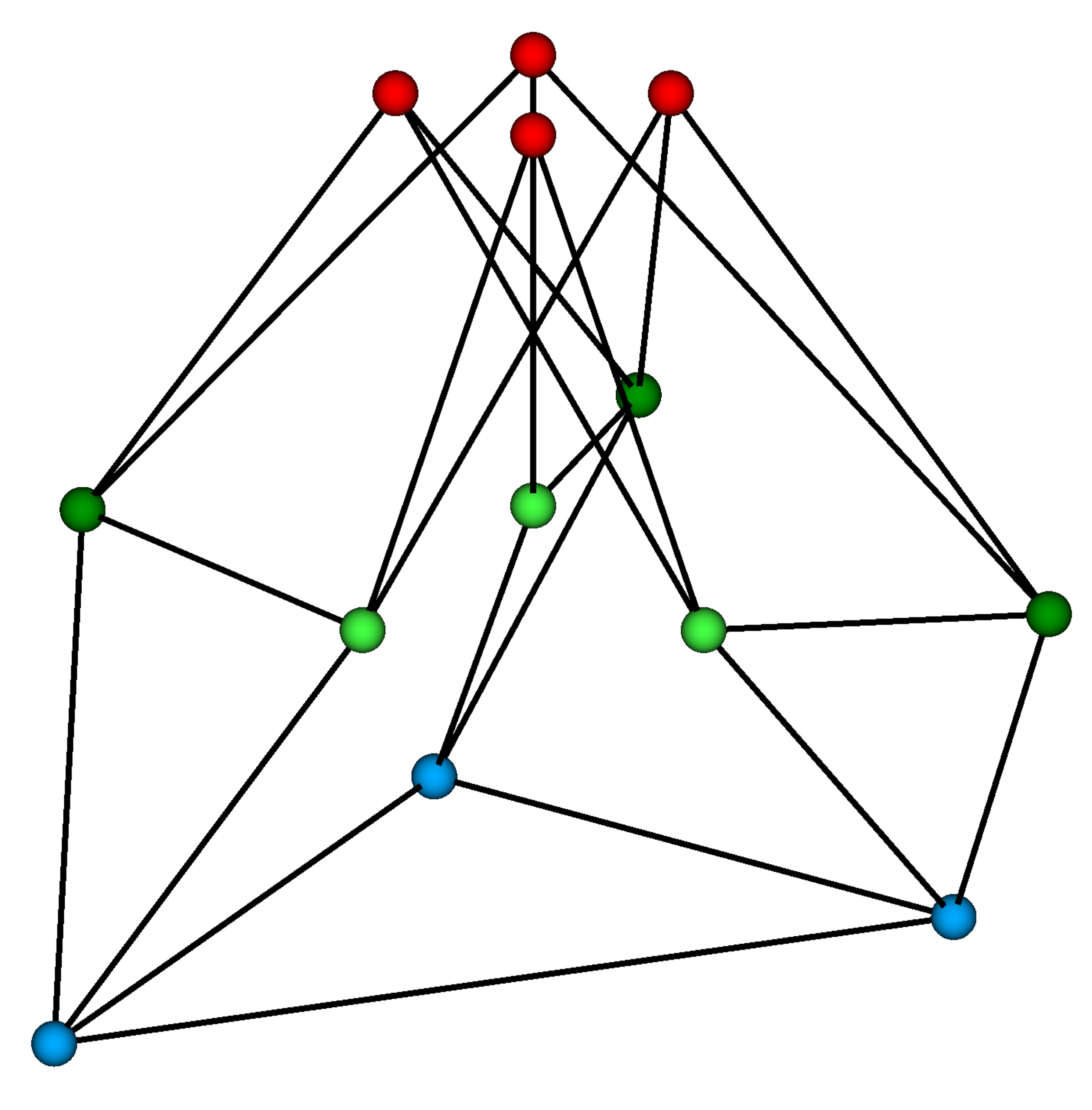}
      \caption{The structure graph $G$: $ER_3$.}
      \label{fig:star_er}
    \end{subfigure}%
    \begin{subfigure}{.4\linewidth}
      \centering
      \includegraphics[width=.7\linewidth]{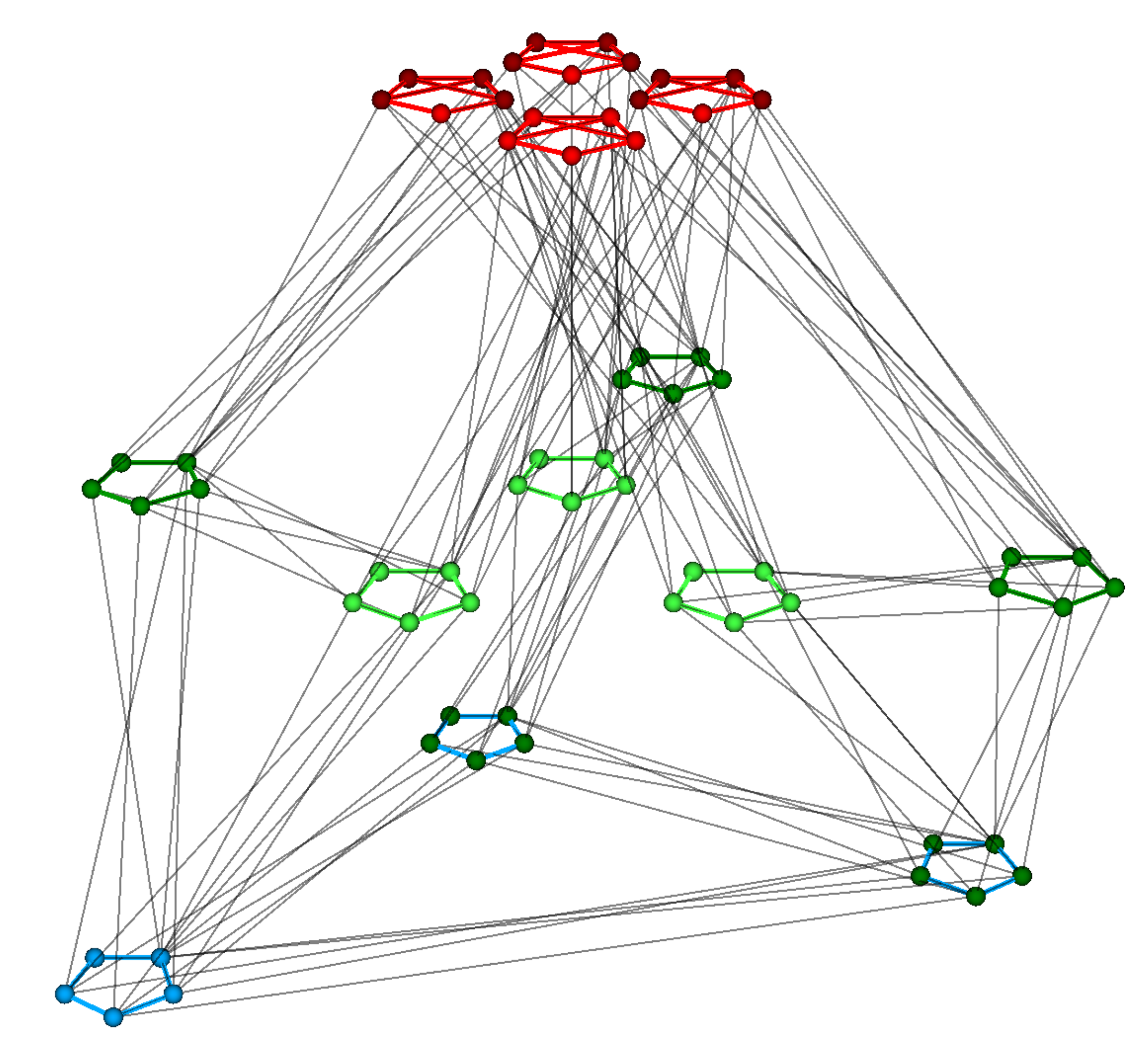}
      \caption{$G*G'$: $ER_3*Paley($5$)$. }
      \label{fig:star_star}
    \end{subfigure}
    \begin{subfigure}{.75\linewidth}
      \centering
      \includegraphics[width=.7\linewidth]{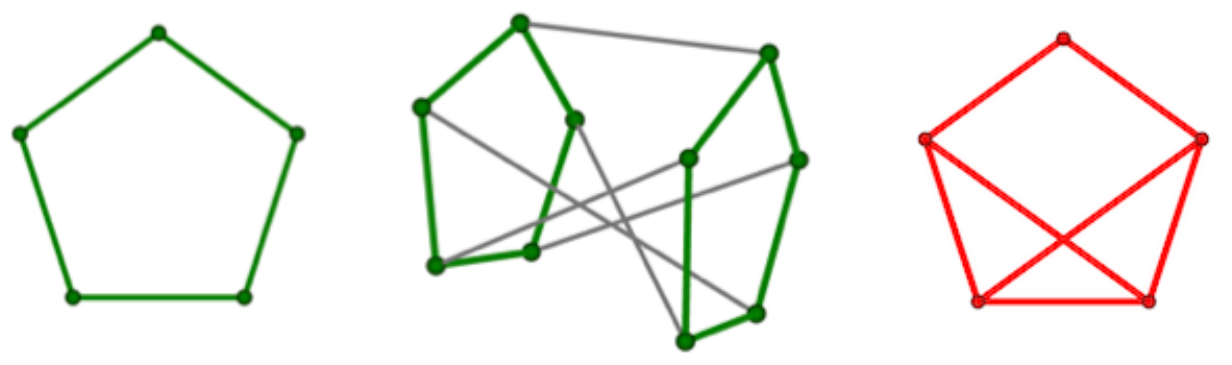}
      
      \caption{The supernode $G'$: Paley($5$). We show a non-self-loop supernode (green) and the connections between two non-loop supernodes (green). $ER_3$ happens to have vertices with self-loops. In that case, the star product adds edges to those supernodes, as shown in red.
      }
      \label{fig:star_paley}
    \end{subfigure}
\caption{Construction of the star product $ER_3*Paley($5$)$.
}
\label{fig:star_ps}
\Description{Construction of our star product}
\end{figure}

For the supernode $G'$, 
we construct a new family of graphs, called \newSuper ($IQ$). This family, where it exists, has Property $R^*$ and meets the upper bound on the order of a supernode.
We also mention the Paley graph \cite{Paley1933OnOM, Erdos1963AsymmetricG} as a $G'$ candidate. It is only slightly smaller, and has symmetricity, useful in network design.
The Inductive-Quad and Paley graphs give the largest star-products for almost all radixes: the $IQ$ supernode almost always produces the largest star product, and Paley produces the largest of the others. We show the star-product construction of PolarStar in Figure~\ref{fig:star_ps}.

\begin{table}[ht]
\setlength{\tabcolsep}{2.5pt}
\centering
\footnotesize
\resizebox{\linewidth}{!}{
\begin{tabular}{lllccc@{}}
\toprule
\textbf{Supernodes} & 
\makecell[l]{\textbf{Order}} & 
\makecell[l]{\textbf{Permitted $d'$}} & 
\makecell[l]{\textbf{Symmetric}} &
\makecell[l]{{$R^*$}} & 
\makecell[l]{{$R_1$}} 
\\
\midrule
\newSuper & $2d'+2$ & $0$ or $3$ (mod $4$) & N & Y & N \\
Paley~\cite{Paley1933OnOM, Erdos1963AsymmetricG} & $2d'+1$ & {even, 2$d'+1$ a prime power} & Y & N & Y \\ 
BDF~\cite{bermond82} & $2d'$ & all & N & Y & N \\ 
Cayley \cite{mckay98} & $2d'+\delta, \delta \in \{0,\pm 1\}$ &  $2d'+\delta$ a prime power & Y & N & Y \\ 
Complete & $d'+1$ & all & Y & Y & Y \\
\bottomrule
\end{tabular}
}
\caption{Comparison of parameters of degree $d'$ supernodes.}
\label{table:supernode_chart}
\end{table} 

Other supernode topologies may be of interest, so we mention these here for completeness. For instance, 
complete graphs provide densely-connected
regions of locality, and Cayley graphs are highly symmetric~\cite{mckay98}. The BDF graphs are designed in~\cite{bermond82} specifically for large star products.  We list supernode choices in 
Table~\ref{table:supernode_chart}, and show the orders of $IQ$ and Paley for relevant radixes in Figure~\ref{fig:design_space}.
\subsubsection{Construction of \newSuper Graphs} \label{sec:Quad}
We show here a new construction of a degree-$d'$ \ref{prop_R_star} graph that \emph{reaches the upper bound} $2d' + 2$ on order of \ref{prop_R_star} graphs. 

We begin with two base \ref{prop_R_star} Inductive-Quad graphs shown in Figure~\ref{fig:inductive_base}: $IQ_0$ of degree $d'=0$ with $2$ vertices, and $IQ_3$ of degree $d'=3$ with $8$ vertices. Both have $2d'+2$ vertices.  

We then inductively construct a graph of degree $d'+4$ with Property~\ref{prop_R_star} from a graph $IQ_{d'}$ of degree $d'$ with Property~\ref{prop_R_star} that has $2d'+2$ vertices. As shown in Figure~\ref{fig:inductive_increase}, $V(IQ_{d'})$ may be 
partitioned into two disjoint sets of $d+1$ vertices each: $A$ and $f(A)$. 
To construct $IQ_{d'+4}$,
we add the $8$ vertices of $IQ_3$:
$$\{x',f(x'),y',f(y'),z',f(z'),w',f(w')\}$$ 
along with the $IQ_3$ edges. Next, we add the following edges:
\begin{itemize}[itemsep=0pt,parsep=2pt]
    \item between $\{x',f(x'),z',f(z')\}$ and all vertices in $A$, and
    \item between $\{y',f(y'),w',f(w')\}$ and all vertices in $f(A)$.
\end{itemize}
\begin{figure}[htbp]
    \centering
    \begin{subfigure}{0.9\linewidth}
      \centering
      \includegraphics[width=.7\linewidth]{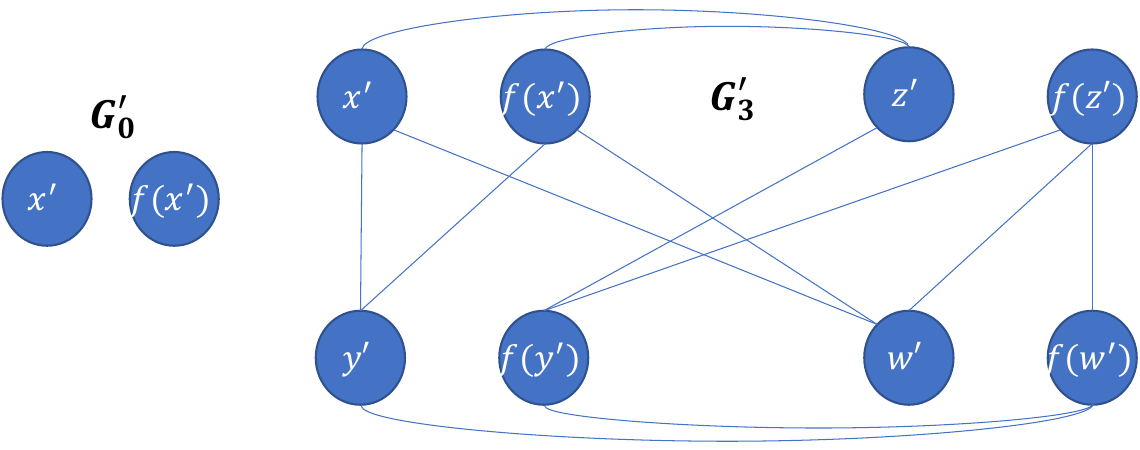}
      \caption{Base \newSuper graphs of degree $d'=0$ and $d'=3$.}
      \label{fig:inductive_base}
    \end{subfigure}
    \begin{subfigure}{\linewidth}
      \centering
      \includegraphics[width=.7\linewidth]{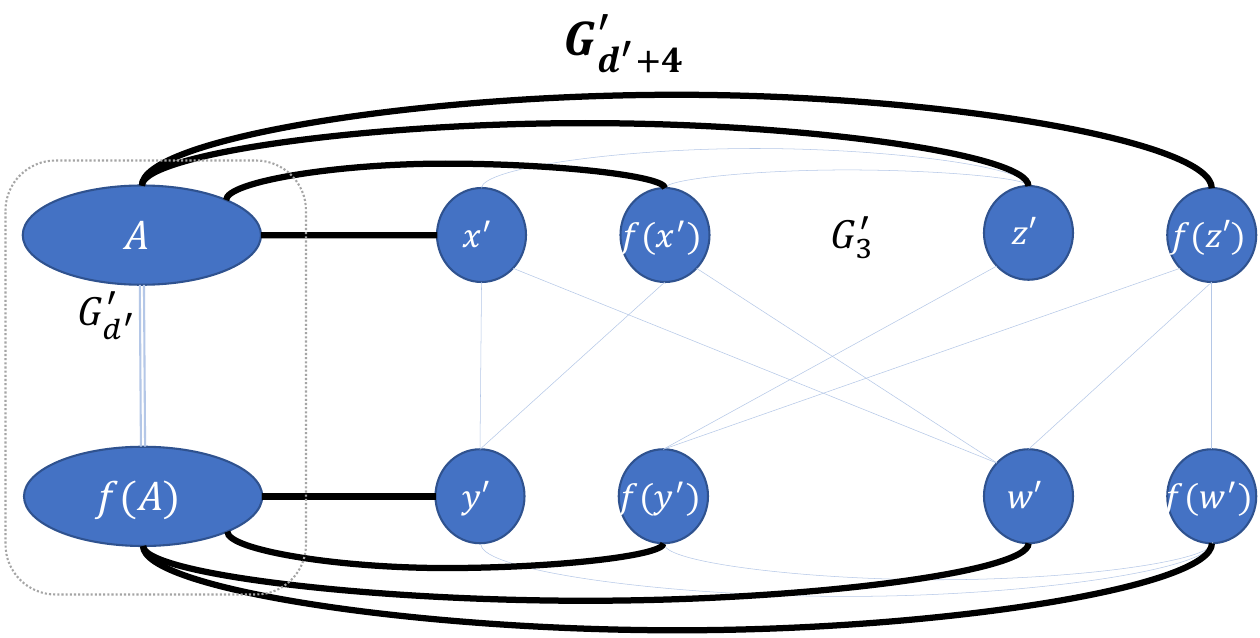}
      \caption{The construction of the \newSuper graph of degree $d'+4$ from an Inductive-Quad graph of degree $d'$ and $IQ_{3}$.}
      \label{fig:inductive_increase}
    \end{subfigure}
\caption{Inductive construction of \newSuper topology with embedded bijection $f$ that satisfies Property~\ref{prop_R_star}.}\label{fig:inductive_quad}
\Description{}
\end{figure}
\begin{proposition}\label{prop:IQ_propRstar}
    $IQ_{d'}$ has $2d'+2$ vertices and degree $d'=4n$ or $d'=4n+3$, and has Property~\ref{prop_R_star}.
\end{proposition}
\begin{proof}
    The order and degree follow from the construction. Property~\ref{prop_R_star} follows by induction, noting by inspection that $IQ_0$ and $IQ_3$ both have Property~\ref{prop_R_star}, and noting that the inductive construction then guarantees Property~\ref{prop_R_star}. 
\end{proof}
\begin{corollary}\label{propose:existence}
 For every integer $n\geq 0$, there exists a graph $G'_{d'}$ of degree $d'=4n$ or $d'=4n+3$ that satisfies $R^*$ and has $2d'+2$ vertices.
 \end{corollary}
Corollary~\ref{propose:existence} implies that the $IQ$ graphs are larger than those discussed in \cite{bermond82}, as those graphs have a smaller $G'$ supernode.

It is also true that graphs $G'_{d'}$ satisfying $R^*$ and having $2d'+2$ vertices exist only for $d'=4n$ or $d'=4n+3$ for some non-negative integer $n$, but we do not prove this here.

\section{Design Space of PolarStar}
We evaluate the scale of network achievable by \newfly and compare it against existing diameter-3 topologies.

\subsection{Theoretical Analysis}
Recall the degree of
star product $G*G'$ is $\deg(G)+\deg(G')$, and the order is $\abs{V(G)}\cdot\abs{V(G')}$. Our structure graph is $ER_q$, which has degree $d=q+1$ and order $q^2+q+1$, where $q$ is a prime power. Using a
\newSuper supernode of degree
$d'$, we get a \newfly $G_*$ of
degree $d_*=d+d'$ and order
$$
\abs{V(G_*)}= (q^2+q+1)(2d'+2) = (q^2+q+1)(2d_* - 2q).
$$
The order is maximized for
\begin{equation}
    \arg\max_q{V(G_*)} = \frac{(d_*-1) + \sqrt{(d_*-1)(d_*-2)}}{3} \approx \frac{2d_*}{3}.\label{eq:argmax_q}
\end{equation}
Substituting this value of $q$, we get that the maximum order of \newfly for a given degree $d_*$ is
\begin{equation}\label{eq:polarstar_order}
\max_{\text{\newSuper}}{\abs{V(G_*)}} \approx \frac{8d_*^3 + 12d_*^2 + 18d_*}{27}.
\end{equation}
Similarly,
$\max_{\text{Paley}}{\abs{V(G_*)}}$ $\approx \frac{8d_*^3}{27}$.
Thus, \newfly  asymptotically reaches $\frac{8}{27}^{\text{th}}$ Moore bound for diameter-3 graphs.

In practice, the degree distribution among factor graphs is constrained as (a)~$q$ must be a prime power in an ER graph of degree $q+1$, and (b)~\newSuper and Paley graphs exist for a subset of integer radixes, as shown in Table~\ref{table:supernode_chart}. Therefore, we evaluate all feasible combinations of $d$ and $d'=d_*-d$ for both \newSuper and Paley supernodes to find the largest configuration.
\subsection{Scalability in Practice}
Figure~\ref{fig:MB_comparison} compares the scalability
of \newfly and other direct diameter-3 topologies, in terms
of their Moore-bound efficiency.
Clearly, \newfly \emph{exceeds the scalability} of
all known diameter-3 topologies. Compared to HyperX~\cite{ahn2009hyperx} and Dragonfly~\cite{dally08}, it 
achieves $6.7\times$ and $1.9\times$ 
geometric mean increase in the order, respectively.
\begin{figure}[htbp]
    \centering
    \includegraphics[width=0.95\linewidth]{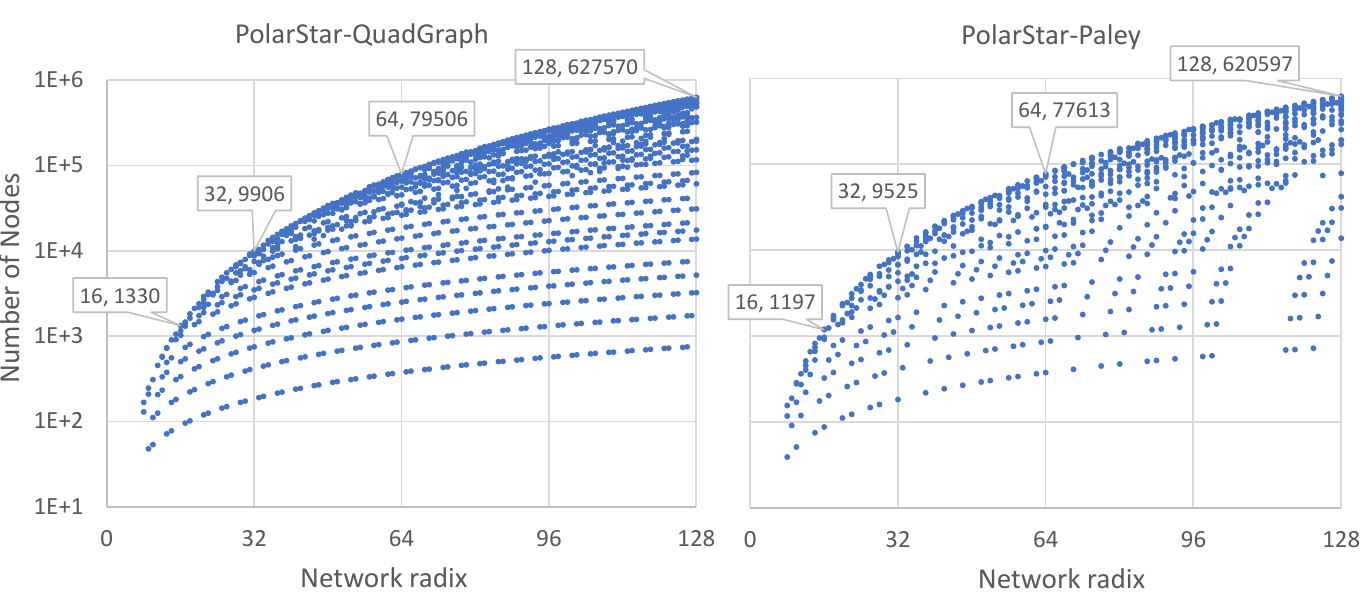}
    \caption{Feasible combinations of radix and order in \newflyN.}
    \label{fig:design_space}
    \Description{}
\end{figure}

Unlike the state-of-the-art
Bundlefly~\cite{lei2020bundlefly}, \newfly offers a 
construction for every network radix and 
more consistent scaling with respect to the Moore bound. Overall, it is $1.3\times$ geometric mean larger than Bundlefly, which results from
the use of more scalable structure graphs and supernodes in its star product, with more possibilities of degree distributions for better 
optimization of \looseness=-1scale.

\newfly 
also approaches the optimal scale for a star-product network with the currently known factor graph properties that give a diameter-3 product. This is because
(a)~the Erd\H os-R\'enyi structure graph asymptotically reaches the
diameter-2 Moore bound, and
(b)~the \newSuper supernode topology reaches the
optimal order for graphs satisfying Property~\ref{prop_R_star}. 

For all radixes, the largest \newfly order for 
degrees $k\in[8,128]$ is 
constructed with the \newSuper supernode, except $k={23, 50, 56, 80}$, where Paley supernode gives a larger topology.  
\newfly also offers a wide range of network orders for each radix, as shown in 
Figure~\ref{fig:design_space}. This diversity of feasible
designs is enabled by varying (a)~the degree distribution between 
structure and supernode graphs, and (b)~the choice of
supernode \looseness=-1graph.

\section{Layout}
\begin{figure*}[htbp]
\centering
\begin{subfigure}[t]{0.31\textwidth}
\centering
\includegraphics[width=0.6\columnwidth]{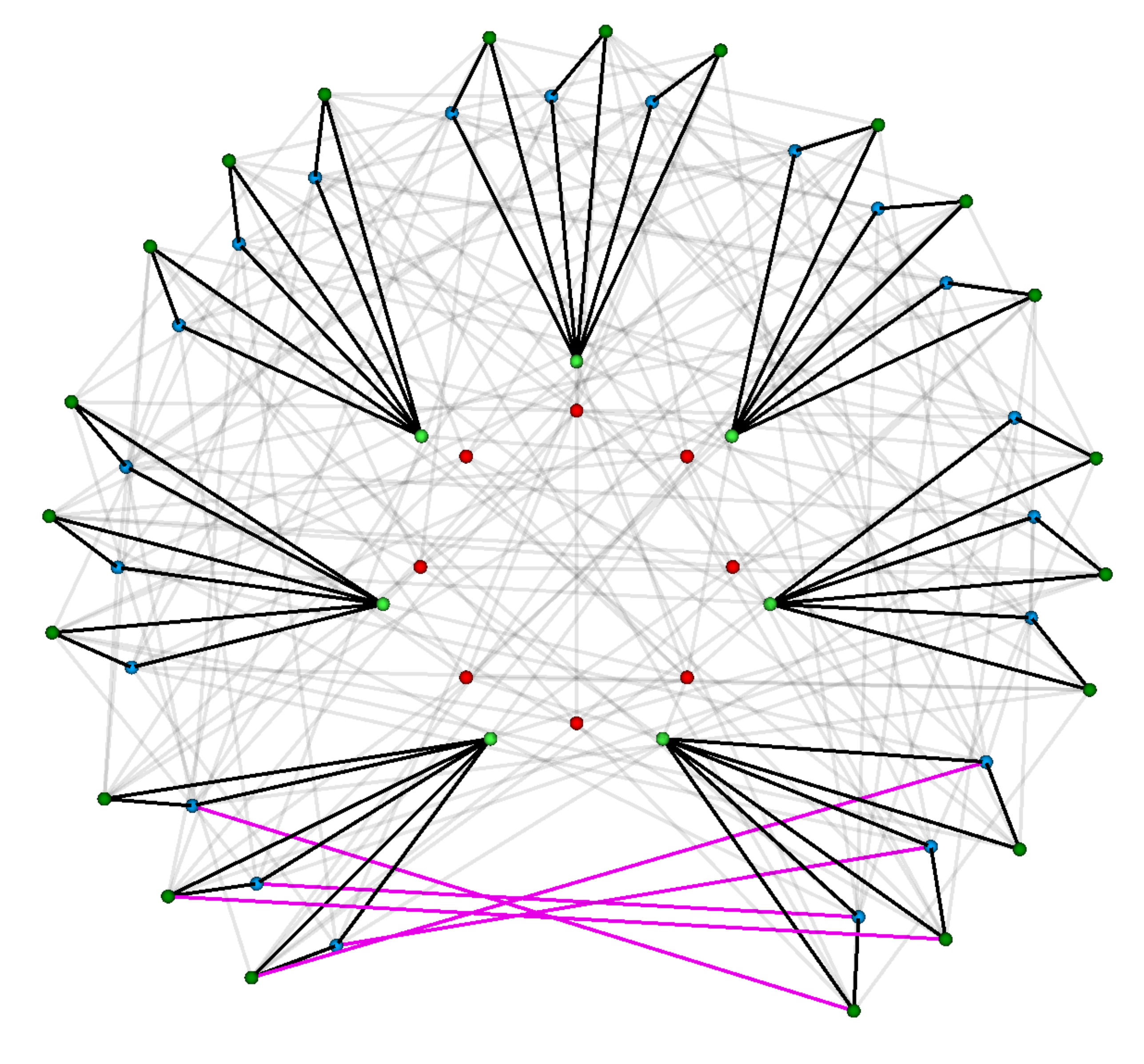}
\caption{Modular layout for $ER_7$ graph~\cite{polarfly_sc22}. Each group of $3$ triangles with a common node is a non-quadric cluster. Red nodes form a quadric cluster. Magenta edges connect two clusters.}
\label{fig:er7_layout}
\end{subfigure}\hspace{3mm}
~
\begin{subfigure}[t]{0.31\textwidth}
\centering
\includegraphics[width=0.6\columnwidth]{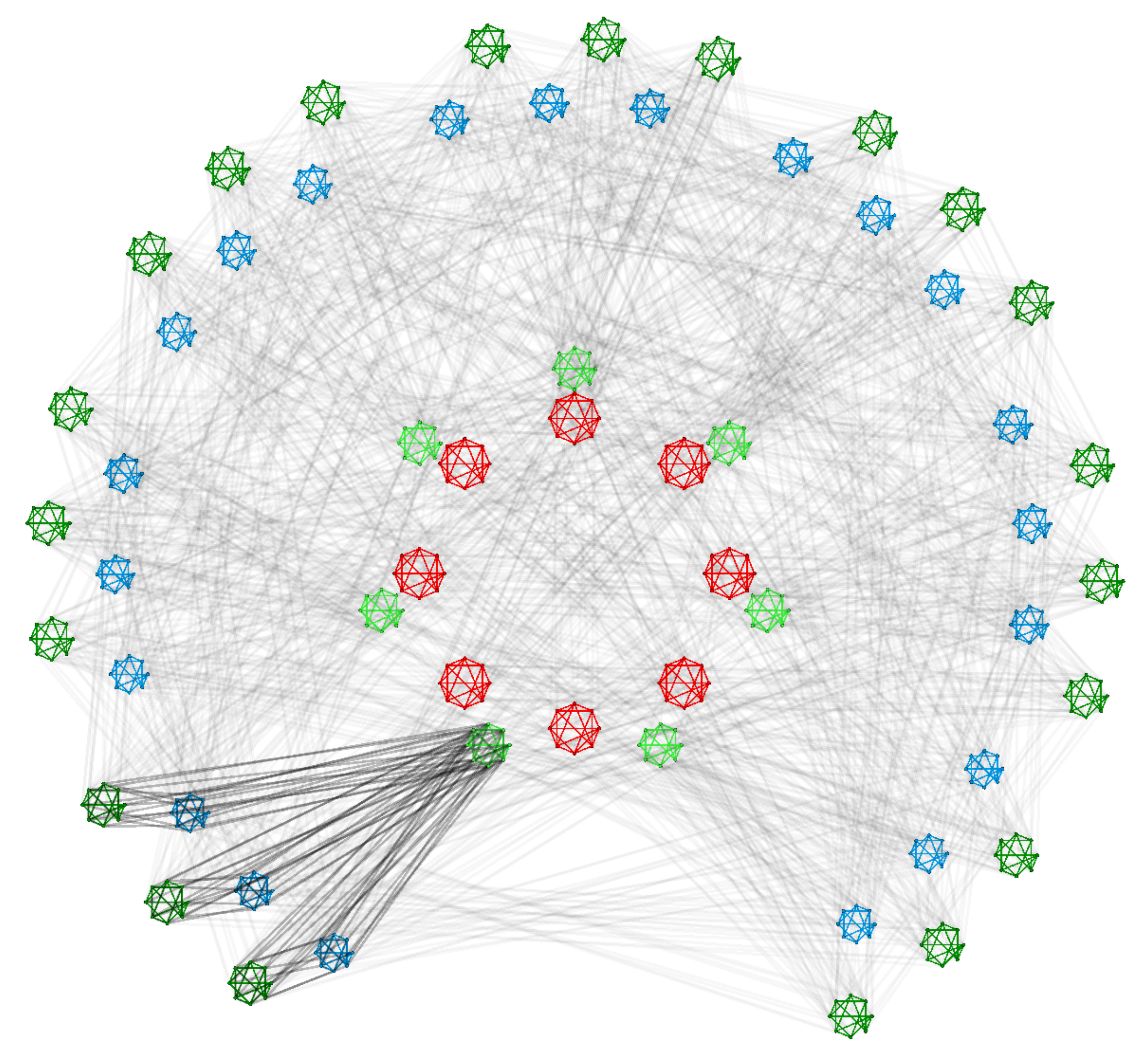}
\caption{Layout of PolarStar$_{11}$ with an $ER_7$ structure graph. Each $ER_7$ node becomes an $IQ_3$ supernode in \newflyN. The highlighted links and incident supernodes are a supernode cluster.}
\label{fig:polarstar_inter_supernode}
\end{subfigure}\hspace{3mm}
~
\begin{subfigure}[t]{0.31\textwidth}
\centering
\includegraphics[width=0.6\columnwidth]{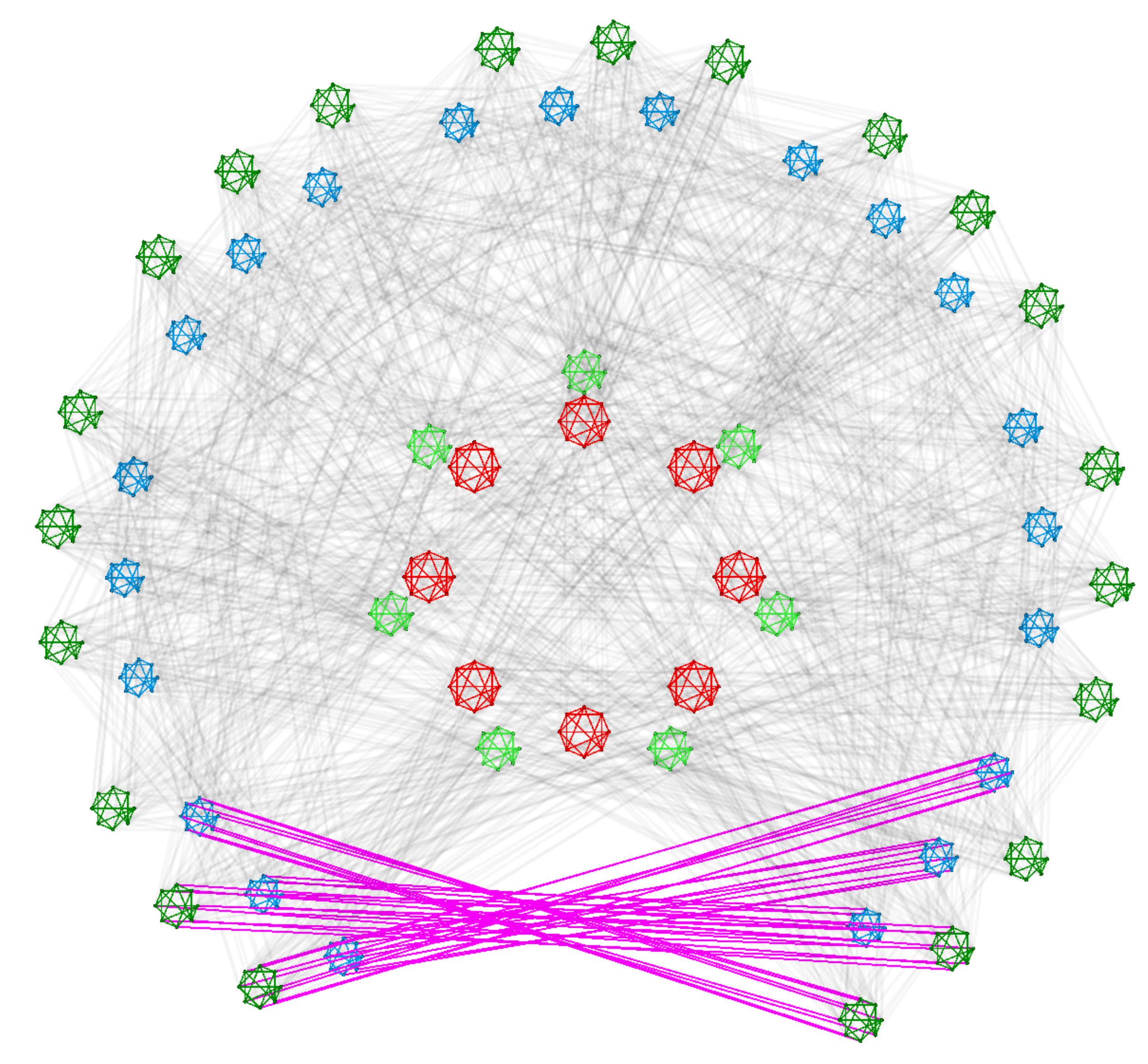}
\caption{Each pair of supernode clusters in \newfly are connected by multiple link bundles in magenta. Each link bundle corresponds to a single inter-cluster link of $ER_7$ shown in Figure~\ref{fig:er7_layout}.}
\label{fig:polarstar_inter_cluster}
\end{subfigure}
\caption{Hierarchical Modular Layout for \newfly derived from a layout for ER structure graphs used in the PolarFly network~\cite{polarfly_sc22}. Adjacent supernodes are
connected by a bundle of links and adjacent supernode clusters are connected by multiple such bundles.}
\label{fig:polarstar_layout}
\Description{}
\end{figure*}

Physical deployment favors modular topologies
comprised of smaller identical subgraphs that
may be implemented as blades or racks.
Also, if multiple links
connect adjacent modules, they can
be bundled into Multicore Fibers~(MCFs) to reduce cable count~\cite{awaji2013optical, lei2020bundlefly}. 
We explain the hierarchically modular structure of \newfly shown in Figure~\ref{fig:polarstar_layout} and point to opportunities for bundling links in \newflyN. 

Consider a maximum order \newfly of degree $d_*$ with $ER_q$ structure graph of degree $q+1\approx \frac{2d_*}{3}$~(Equation~\eqref{eq:argmax_q}), and \newSuper supernode of degree $d'=d_*-(q+1)$.
In this \newflyN, the \emph{supernode} is 
the smallest building block~(Figure~\ref{fig:star_ps}) of $2d_*-2q$ nodes, and is replicated $q^2+q+1$ times~(once per node of $ER_q$). 
There are $2(d_*-q)$ links between each pair of adjacent supernodes. If these can be bundled in an MCF, we get 
$q(q+1)^2$ inter-module MCFs (same as the non-self-loop edges in $ER_q$~\cite{polarfly_sc22}). Therefore, such bundling reduces the global cables by a factor of $\approx\frac{2d_*}{3}$.

The next level in modular hierarchy is the \emph{supernode clusters}. As shown in~\cite{polarfly_sc22}, $ER_q$ graph can be divided into $q+1$ clusters with approximately $q$ links between each pair of clusters.
In \newflyN, this translates to  $q+1$ supernode clusters with approximately $q$ bundles of links between each \looseness=-1pair as shown in Figure~\ref{fig:polarstar_inter_cluster}. 

\section{Evaluation: Synthetic 
Patterns
}\label{sec:eval}
\subsection{Topologies}\label{sec:baselines}
We compare \newfly with  Bundlefly~\cite{lei2020bundlefly}, {Megafly~\cite{flajslik2018megafly, shpiner2017dragonfly+} and Spectralfly~\cite{aksoy2021spectralfly}} as state-of-the-art diameter-3 networks, 
3-D HyperX~\cite{ahn2009hyperx} and Dragonfly~\cite{dally08} as popular diameter-3 networks in practice, and 3-level Fat-trees~\cite{Leiserson:1985:FUN:4492.4495} as the most widely used network, as is standard, using Booksim simulator~\cite{jiang2013detailed}. 

Networks such as torus, hypercube or Flattened Butterfly 
have been shown to have lower performance than these baselines~\cite{besta2014slim, dally08}. {We also explored the Galaxyfly family of flexible low-diameter topologies~\cite{lei2016galaxyfly}. A diameter-3 Galaxyfly is isomorphic to a Dragonfly, which is included in the comparison.}

\begin{table}[ht]
\setlength{\tabcolsep}{2.5pt}
\centering
\footnotesize
\resizebox{\linewidth}{!}{
\begin{tabular}{llcccc@{}}
\toprule
%

\textbf{Network} & \textbf{Parameters} & \textbf{\# Routers} & \textbf{\begin{tabular}[c]{@{}c@{}}Network\\ Radix\end{tabular}} & \textbf{\# Endpoints}\\

\midrule
\newfly with \newSuper(PS-IQ) & $d\text{=12},\ d'\text{=3}$, p=5        & 1,064       & 15 & 5,320 \\
\newfly with Paley~(PS-Pal) & $d\text{=9},\ d'\text{=6}$, p=5        & 993       & 15 & 4,965 \\
Bundlefly~(BF) & $d\text{=11},\ d'\text{=4}$, p=5 & 882 & 15 & 4,410\\
3-D HyperX~(HX) & $9\times 9\times 8$ , p=8 & 648 & 23 & 5,184\\
Dragonfly (DF) & a=12, h=6, p=6     & 876       & 17 & 5,256 \\
Spectralfly~(SF) & $\rho$=23, $q$=13, p=8 & 1,092 & 24 & 8,736 \\
Megafly~(MF) & $\rho$=8, $a$=16, p=8 & 1,040 & 16 & 4,160 \\
3-level Fat-tree (FT) & n=3, p=18        & 972       & 36 & 5,832\\
\bottomrule
\end{tabular}
}
\caption{Simulated configurations of topologies. Endpoints per router $=p$.}
\label{table:config}
\end{table} 



The configurations of topologies used are shown in Table~\ref{table:config}. 
For direct diameter-3 networks, $p$, the number of endpoints per router, is $1/4^{\text{th}}$ of total ports ($1/3^{\text{rd}}$ of the network radix). 
{For Fat-tree and Dragonfly, we use the built-in BookSim constructions and routing~\cite{jiang2013detailed}. Booksim's Fat-tree topology for router radix $2p$ and $3$ layers has $p^2=324$ routers in each layer}, with top layer routers having half the radix~($18$). This configuration supports $p^3=5,832$ endpoints.
In Megafly and Fat-tree, half and one-third of the routers, respectively, have endpoints on half of their ports.
\subsection{Minimal Path Computation in PolarStar}\label{sec:minpath}
We use the properties of structure graph $G$ and supernode $G'$ to compute each minimal path (minpath) in PolarStar. This significantly reduces the storage
requirements of the router compared to table-based routing which stores output port along minpath towards every router in the network. 

The structure graph in PolarStar is an \emph{Erd\H os-R\'enyi} polarity graph with diameter $D=2$ and has Property~\ref{prop_R}.
We describe minpaths in Inductive-Quad supernode that satisfies Property~\ref{prop_R_star}; computation for a Paley supernode is similar. 
Our minpath computations are amenable to incremental routing and therefore, suitable for \emph{destination-based routing}.

If the source and destination are adjacent, minpath is just the edge between them. 
Next, we show that when source and destination are in adjacent supernodes, there is a 2-hop minpath between them. These paths are graphically illustrated in Figure~\ref{fig:diameter_theorem} when routing from $(a_1,f(x'))$ to $(y,y')$.
For clarity and alignment with Figure~\ref{fig:diameter_theorem}, we assume $(a_1,f(x'))$ is the source, and $(y,y')$ is the destination.
We compute the minimal paths in this case by individually evaluating all scenarios of Property~\ref{prop_R_star}:
\begin{enumerate}[label=(\alph*),itemsep=0pt,parsep=2pt]
    \item $y'=x' \rightarrow$ $(a_1,f(x'))$ and $(y,x')$ are adjacent.
    \item $y'=f(x')\rightarrow$ We do not hop to supernode $y$ directly from the source. 
    From Property~\ref{prop_R}, there is a 2-hop path $(a_1,b_1,y)$ in structure graph $G$.
    We take the corresponding path in $G_*$ which is $\left((a_1,f(x')), (b_1,x'), (y,f(x'))\right)$, see Figure~\ref{fig:diameter_p1}. 

    \item $(y',x')\in E(G')\rightarrow$ we use the path \newline
    $\left((a_1,f(x')), (y,x'), (y,y')\right)$, as shown in Figure~\ref{fig:diameter_p2}.

    \item $(f(y'),f(x'))\in E(G')\rightarrow$ we use the path \newline$\left((a_1,f(x')), (a_1,f(y')), (y,y')\right)$, as shown in Figure~\ref{fig:diameter_p3}.   
\end{enumerate}

Lastly, we consider the case when source and destination are neither adjacent nor in adjacent supernodes. 
This is also shown in Figure~\ref{fig:diameter_theorem} where source is $(x,x')$ and destination is $(y,y')$.
From Property~\ref{prop_R}, there exists a 2-hop path $(x,a_1,y)$ in $G$. In $G_*$, we always hop from $(x,x')$ to $(a_1,f(x'))$. 
Since $(a_1,f(x'))$ is in a supernode adjacent to $(y,y')$, we then take the 2-hop path to $(y,y')$ as described above, giving an overall path of length at most 3.

\subsection{Routing}\label{sec:routing}
We use the following well-known routing schemes:
\begin{itemize}[itemsep=0pt,parsep=2pt,leftmargin=*]
    \item \emph{Minimal Routing~(MIN):} Each packet between the source and destination is routed along a shortest path. In MF, we use the path diversity between routers within the same supernode~(group).
    We saw poor performance on SF and BF when using a single minpath per router pair, so for these we store all minpaths for every destination in a large routing table. 
   HX also uses all minpaths but computes them by aligning coordinates in each dimension, so
   does not require large routing tables~\cite{ahn2009hyperx}.
   For PolarStar, minpath is  
    analytically computed as described in Section~\ref{sec:minpath}. 
    It only stores shortest paths in the structure graph and requires significantly less memory compared to SF and BF.   
        
    \item \emph{Load-balancing Adaptive Routing~(UGAL):} 
    Valiant routing balances the load by misrouting each packet
    via a randomly chosen intermediate router. In our 
    implementation, we sample
    $4$ intermediate routers at random for misrouting. We predict path latency for valiant paths and minpaths using local buffer occupancy, and select the smallest latency path. 
    Note that UGAL routing in SF, BF and HX uses multiple minpaths stored in routing tables.
\end{itemize}

\subsection{Simulation Setup}\label{sec:simulation}
We analyze network performance using the cycle-accurate BookSim simulator~\cite{jiang2013detailed}. {Simulation parameters~(latency, bandwidth) are normalized to the values of a single link.} Credit-based flow-control is used to create backpressure and restrict injection for congestion-control.
Packets are of size $4$ flits and 
input-queued routers have $128$ flit buffers per port and $4$ virtual channels~(VCs). {Dragonfly and Megafly respectively use $2$ and $1$ VCs for MIN routing, and $3$ and $2$ VCs for UGAL routing.}
A warm-up phase precedes simulations so the network may reach steady-state before measurements.

{We analyze network performance for synthetic traffic patterns that represent crucial applications.
Such patterns are widely used to evaluate network topologies~\cite{dally08, lei2020bundlefly, polarfly_sc22, ahn2009hyperx, besta2014slim}.}

\begin{enumerate}[itemsep=0pt,parsep=2pt,leftmargin=*]
    \item \emph{Uniform} random traffic -- the destination of each packet is selected uniformly at random~(represents graph processing
and irregular algorithms e.g. sparse linear algebra~\cite{dang2018lightweight, besta2015accelerating}).

\item \emph{Random permutation} traffic -- a fixed permutation 
mapping $\tau$ of source to destination \emph{routers} is chosen uniformly at random. 
Endpoints on a router $R_s$ transmit only to corresponding endpoints on router
$\tau(R_s)$. This pattern also emulates permutation traffic
for co-packaged nodes with both compute and router.
{Permutation traffic is commonly seen in FFT, physics simulations and collectives~\cite{besta2014slim}. 
The random permutation pattern represents traffic generated by these applications when process IDs are randomly mapped to nodes.

\item \emph{Bit Shuffle} traffic -  the destination address 
is obtained by shifting the source address bits to the left by 1~($d_i = s_{(i-1) \mod b}$). This pattern is common in FFT and sorting algorithms~\cite{bahn2008generic}.
\item \emph{Bit Reverse} traffic - the destination address 
is obtained by reversing the bit order of source address~($d_i = s_{b - i - 1}$). This pattern occurs in Cooley-Tukey FFT, binary search and dynamic tree data structures~\cite{wilber1989lower, sarbazi2001communication, gold1969digital}.}
\end{enumerate}
Bit Shuffle and Bit Reverse traffic uses $2^b$ endpoints, where $2^b$ is the largest power of two no more than the total endpoints.
Endpoint IDs are contiguous for each router and supernode/group in hierarchical topologies~(\newflyN, Bundlefly, Dragonfly, Megafly, Fat-tree). In such topologies, almost all endpoints in a supernode communicate with only two other supernodes under Bit Shuffle traffic.
\subsection{Results}\label{sec:perf_results}
Figure~\ref{fig:topologie_comparison} compares \newfly performance against the baseline topologies, for different routing schemes and traffic patterns. Labels follow the scheme $\texttt{<topology>-<routing>}$. Load is normalized to the peak injection bandwidth. 

Overall, PS-Pal and PS-IQ perform competitively for most of the patterns.
With MIN routing, PS-Pal and PS-IQ sustain more than $75\%$ of full injection bandwidth~(load) on uniform traffic. 
SF and BF use all minpaths available and have marginally better performance, but require significantly higher storage for routing tables than PS-*.

\begin{figure}[ht]
    \centering
    \begin{subfigure}{.5\linewidth}
      \centering
      \includegraphics[width=\linewidth]{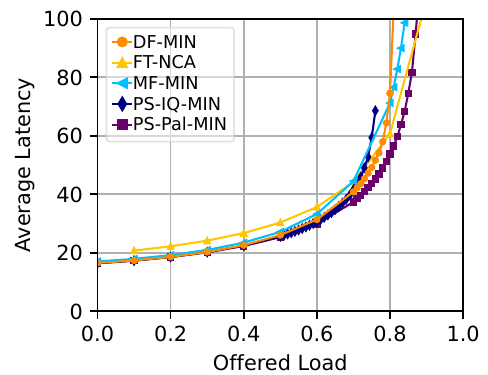}
      \caption{\begin{scriptsize}Uniform traffic, MIN routing.\end{scriptsize}}
      \label{fig:topologie_comparison_sub1}
    \end{subfigure}%
    \begin{subfigure}{.5\linewidth}
      \centering
      \includegraphics[width=\linewidth]{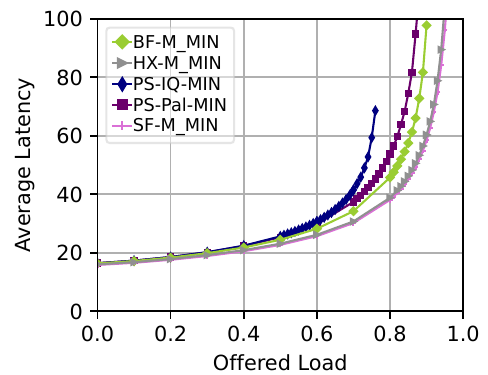}\caption{\begin{scriptsize}
          Uniform traffic, MIN routing.
      \end{scriptsize}}
      \label{fig:topologie_comparison_sub2}
    \end{subfigure}\\%
    \begin{subfigure}{0.5\linewidth}
      \centering
      \includegraphics[width=\linewidth]{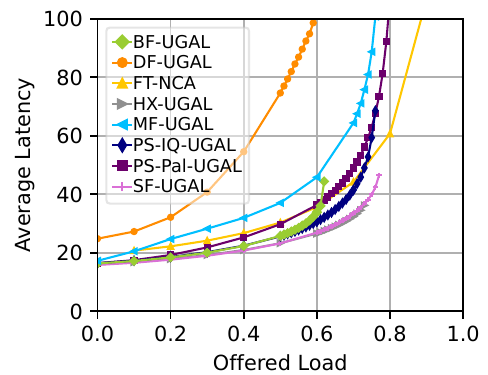}
      \caption{\begin{scriptsize}
          Uniform traffic, UGAL routing.
      \end{scriptsize}}
      \label{fig:topologie_comparison_sub3}
    \end{subfigure}%
    \begin{subfigure}{.5\linewidth}
      \centering
      \includegraphics[width=\linewidth]{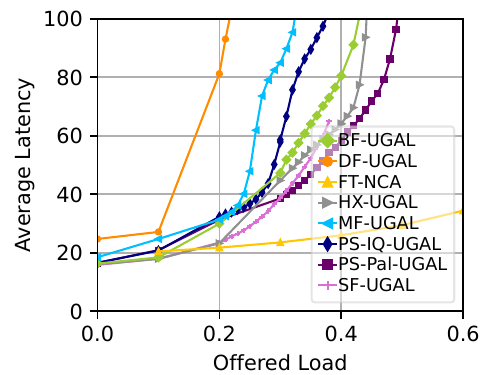}
      \caption{\begin{scriptsize}
          Random Permutation traffic.
      \end{scriptsize}}
      \label{fig:topologie_comparison_sub4}
    \end{subfigure}\\
    \begin{subfigure}{0.5\linewidth}
      \centering
      \includegraphics[width=\linewidth]{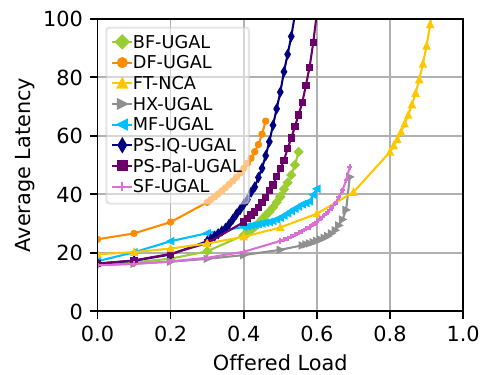}
      \caption{\begin{scriptsize}
      Bit-reverse traffic.\end{scriptsize}}
      \label{fig:topologie_comparison_sub5}
    \end{subfigure}%
    \begin{subfigure}{.5\linewidth}
      \centering
      \includegraphics[width=\linewidth]{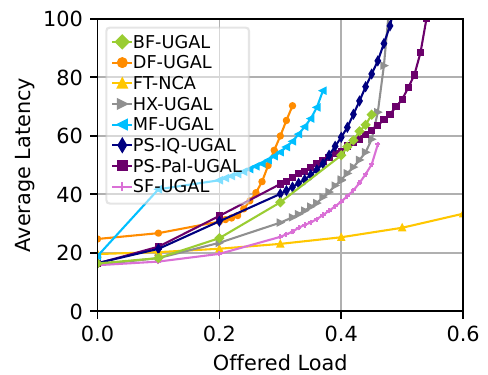}
      \caption{\begin{scriptsize}
          Bit-shuffle traffic.
      \end{scriptsize}}
      \label{fig:topologie_comparison_sub6}
    \end{subfigure}
\caption{Performance comparison of 
several topologies. Latency is measured up to the highest injection rate for which simulation is stable. Beyond this, the network is saturated and 
latency increases with simulation time. 
In SF and BF, each routing table 
stores all minpaths to every destination. Unlike other topologies, they exhibit subpar performance when a single minpath is used.
}
\label{fig:topologie_comparison}
\Description{}
\end{figure}

With adaptive UGAL routing, PS-Pal and PS-IQ sustain between $0.4$ to $0.6$ of the full load on various traffic patterns. Their performance is  significantly better than DF and is comparable to BF (also a star-product), even though BF stores multiple minpaths for UGAL.

{Bit Shuffle performance of MF and DF is poor because they only have one link between each pair of supernodes~\cite{flajslik2018megafly}, as opposed to BF and PS-* which have multiple links between the supernodes. This pattern thus highlights the benefits of star-product topologies over DF and MF. Comparatively, MF performs better on Bit Reverse pattern which has more balanced load distribution. 
}
{For most traffic patterns, HX and SF sustain the highest injection rate. 
They have a high degree of symmetry, and 
high link density. However, they 
lack scaling efficiency~(Figure~\ref{fig:MB_comparison}) and have $6.7\times$ and $12.8\times$ geometric mean lower scale than \newflyN, respectively. 
}
\begin{figure}[ht]
    \centering
    \begin{subfigure}{.5\linewidth}
      \centering
      \includegraphics[width=\linewidth]{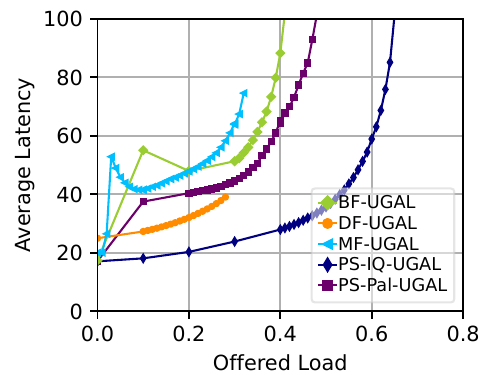}
      \label{fig:rand_endpt_min}
    \end{subfigure}%
    \begin{subfigure}{.5\linewidth}
      \centering
      \includegraphics[width=\linewidth]{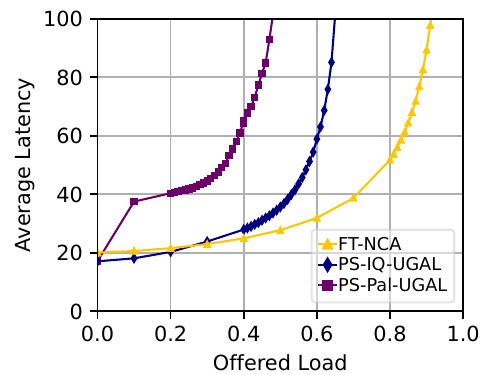}
      \label{fig:rand_endpt_mmin}
    \end{subfigure}
\caption{Performance of  topologies under adversarial traffic.}
\label{fig:adversarial}
\end{figure}
\subsection{Adversarial Traffic}\label{sec:adversarial}
For hierarchical topologies -- PS-*, BF, DF and MF, we implement an adversarial pattern where all endpoints in a supernode transmit to endpoints in only one other supernode, 
resulting in congestion on global~(inter-supernode) links.  
For every source and destination pair, we enforce the longest possible minpath~($3$-hops in PS-*, BF, DF and MF, and $4$-hops in FT), and in PS-* and BF, also the maximum number of global hops~($3$ in PS-*, $2$ in BF).
This traffic pattern has been used for worst-case analysis of BF, DF and MF~\cite{lei2020bundlefly, dally08, flajslik2018megafly}.
This may not be the worst-case pattern for PS-*. The true worst-case may depend on the routing algorithm and the minimum bisection, which is NP-hard to compute for generic graphs.

Figure~\ref{fig:adversarial} shows network performance for this traffic pattern. 
DF and MF saturate at the lowest bandwidth 
with only 
a single link between each supernode pair. 
BF and PS-* perform better because they have multiple links between every supernode pair. PS-IQ performance is better than BF and PS-Pal because of a relatively larger proportion of global links in the particular configuration. 
\section{Evaluation: Real-world Motifs}
\subsection{Simulation Setup}
To compare the performance of topologies on real-world motifs, we use the SST simulator~\cite{rodrigues2011structural} and the network configurations listed in Table~\ref{table:config}, and
construct topologies using SST's Merlin library~\cite{hemmert2018merlin}.

We compare PolarStar (PS-IQ) against the built-in Dragonfly, HyperX and Fat-tree modules in Merlin. 
We evaluate both MIN and adaptive UGAL routing mechanisms for these topologies. For HyperX, Merlin provides DOAL routing which adaptively routes at most once in each dimension. For simulation, we set router buffer size to $64$ KB per port, router and link latencies to $20$ ns and link bandwidth to $4$ GBps.

We use two motifs obtained from the Ember Communication Pattern Library~\cite{hammond2015ember} of SST:
\begin{itemize}[itemsep=0pt,parsep=1pt,leftmargin=*]
    \item \emph{Allreduce} -- An important collective in scientific computing, linear solvers and distributed Machine Learning~\cite{prikopa2016parallel, rabenseifner2004optimization, sergeev2018horovod}. 

    \item \emph{Sweep3D} -- This is a wavefront communication pattern moving diagonally on a 2D processor grid. It stresses network latency and is commonly seen in particle simulations and iterative solvers~\cite{hoisie2007performance}.
\end{itemize}
Message size is $64$KB for Allreduce. For both motifs, we measure time for $10$ iterations. Process IDs are mapped linearly to endpoints.

\subsection{Results}

\begin{figure}[!ht]
    \centering
    \begin{subfigure}{.49\linewidth}
      \centering
      \includegraphics[width=0.9\linewidth]{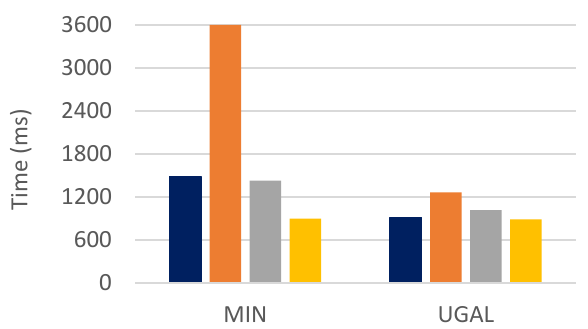}
    \caption{Allreduce Collective}
    \label{fig:sst_allreduce}
    \end{subfigure}
    \begin{subfigure}{.5\linewidth}
      \centering      \includegraphics[width=1.07\linewidth]{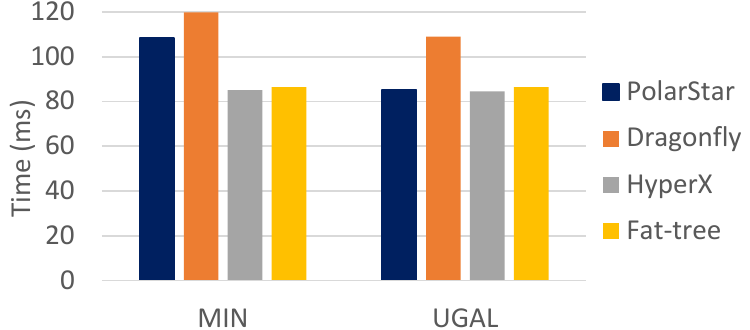}
      \caption{Sweep3D motif}
    \label{fig:sst_sweep}    \end{subfigure}
\caption{Performance of the Allreduce and Sweep3D motifs.}
\label{fig:allred_sweep}
\Description{}
\end{figure}

Allreduce executes fastest on Fat-tree with similar performance on both MIN and adaptive UGAL routing. On the other hand, UGAL performs significantly better than MIN on PolarStar, Dragonfly and HyperX. PolarStar execution time on Adaptive routing is comparable to Fat-tree and better than HyperX. It achieves $2.4\times$ and $1.4\times$ speedup over Dragonfly with MIN and UGAL routing. 
On Sweep3D, PolarStar is marginally faster than Dragonfly with MIN routing. On UGAL routing, PolarStar's performance is similar to HyperX and Fat-tree, and $1.28\times$ better than Dragonfly. The improvement over MIN routing is smaller than that seen on Allreduce. 

Overall, PolarStar performs better than Dragonfly and comparable to Fat-tree and HyperX. However, it has significantly higher Moore-bound efficiency than HyperX~(Figure~\ref{fig:MB_comparison}) and can achieve similar scale with smaller switch radix~(Table~\ref{table:config}).

\section{Structural Analysis}\label{sec:structural_analysis}
\subsection{Bisection Analysis}\label{sec:bisection}
\begin{figure}[!htbp]
    \centering
    \includegraphics[width=0.9\linewidth]{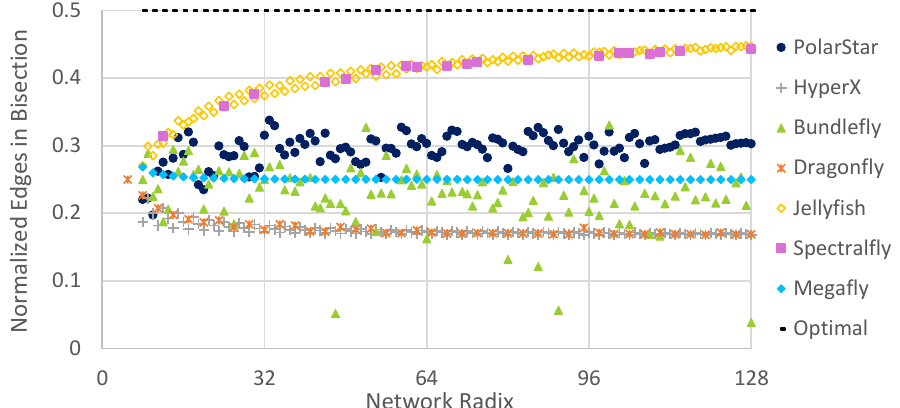}
    \caption{Fraction 
    of links crossing the minimum bisection estimated by METIS~\cite{METIS}.  \newflyN, 
    Bundlefly and Spectralfly use their largest feasible (diameter-3) constructions for each radix.
    Jellyfish has the same radix and scale as \newflyN.
    Fat-tree and Megafly bisection is normalized by 
    the network links incident with routers that have attached endpoints.}
    \label{fig:bisection_baselines}
\Description{}
\end{figure}

Figure~\ref{fig:bisection_baselines} shows the minimum bisection of
different topologies for network radix in range $[8, 128]$. {The
minimum bisection is estimated using METIS~\cite{METIS} for \newflyN, Spectralfly, Megafly, Bundlefly and Dragonfly. We also compare against Jellyfish which uses a random graph topology~\cite{Singla:2012:JND:2228298.2228322}.} 
Among the direct networks, Jellyfish has the highest fraction of links in bisection due
to the random connectivity between vertices, though its diameter is more than $3$. 
{Spectralfly uses Ramanujan graphs that optimize the expansion properties. 
Hence, it has a large bisection~(comparable to Jellyfish), but very few feasible diameter-3 constructions. 
Among the other diameter-3 topologies, \newfly has the \emph{highest} proportion of links crossing the
bisection, with an average of 
$29.6\%$ across all radixes. 
Comparatively, Bundlefly, Dragonfly, HyperX and even the indirect Megafly
only have $22.9\%, 17.8\%, 17.4\%$ and $25.5\%$ links in the bisection cut, respectively.}
The improved bisection cut can be attributed to the
(a)~large bisection of the $ER$ topology of the structure graph~\cite{polarfly_sc22}, and (b)~good radix distribution across 
supernode and structure graphs
due to many choices for supernode radix.

\begin{figure}[htbp]
    \centering
    \includegraphics[width=0.7\linewidth]{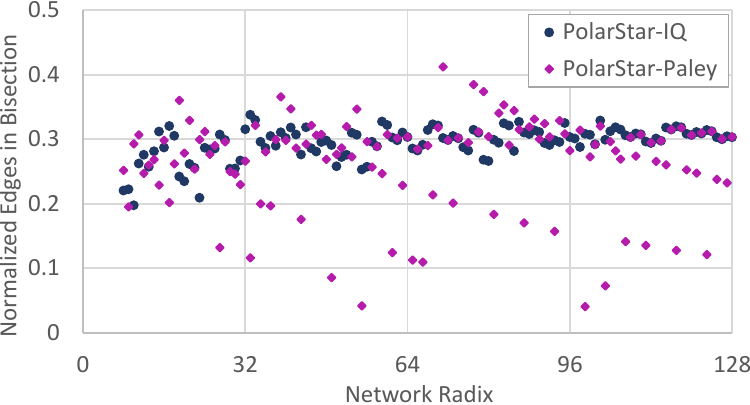}
    \caption{Minimum bisection of \newfly with \newSuper and Paley supernodes, approximated by METIS~\cite{METIS}.
    }
    \label{fig:bisection_polarstar}
    \Description{}
\end{figure}

Figure~\ref{fig:bisection_polarstar}
shows the size of bisection cut of \newfly as a function of radix and
supernode topologies.
\newfly with \newSuper and Paley supernodes have an average 
$29.5\%$ and $26.6\%$ edges in the bisection cut, respectively. The former also offers a more stable bisection across a range of radixes. This is because \newSuper has more 
feasible radixes and allows better distribution of radix between structure graph and supernode. 
Comparatively, the limited choice of radixes for Paley graphs may 
result in a \newfly with large supernodes and small 
structure graphs. Such a network
will have small bisection because most of the links are concentrated within 
dense supernode subgraphs.

\begin{figure}[!htbp]
    \centering
    \includegraphics[width=0.9\linewidth]{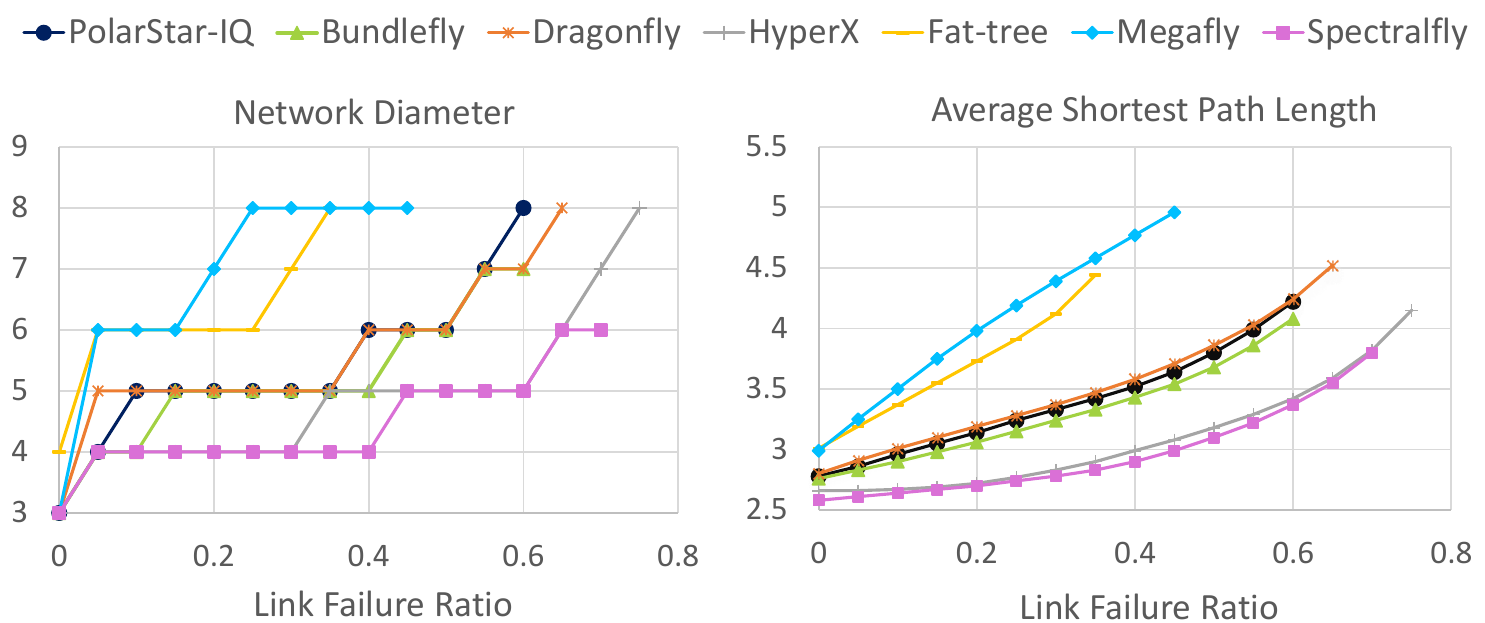}
    \caption{Network Diameter and Average Path Length as a function of random link failures. For Fat-tree and Megafly, we only consider the distance between nodes that have endpoints. PolarStar-Paley behavior is similar to PolarStar-IQ and is not shown for clarity. 
    }    \label{fig:resilience}
    \Description{}
\end{figure}

\subsection{Fault Tolerance}\label{sec:resilience}
To estimate fault tolerance, we simulate
$100$ random link failure scenarios until disconnection.
We randomly select a simulation with median disconnection ratio,
and 
report the variation in network diameter and average shortest 
path length in Figure~\ref{fig:resilience}.

\newfly and Bundlefly have similar
resilience with a
$60\%$ disconnection ratio. Dragonfly has a higher $65\%$ disconnection ratio. However, at low failure ratios, Dragonfly's diameter and average shortest path length increase rapidly. This is likely because if a global
link fails, shortest paths between corresponding groups have to traverse other groups.
HyperX and Spectralfly are the most resilient
of diameter-3 topologies due to higher connection density, 
although they suffer from poor scalablity~(Figure~\ref{fig:MB_comparison}). 
{All direct topologies in Figure~\ref{fig:resilience} have a higher disconnection ratio than the indirect topologies Fat-tree and Megafly. Like Dragonfly, Megafly has only one global link between each pair of groups. Therefore, its diameter increases to $6$ with just $5\%$ failed links, and its average shortest path length increases faster than the Fat-tree.}


\clearpage
\bibliographystyle{ACM-Reference-Format}
\bibliography{references}


\begin{thebibliography}{63}


\ifx \showCODEN    \undefined \def \showCODEN     #1{\unskip}     \fi
\ifx \showDOI      \undefined \def \showDOI       #1{#1}\fi
\ifx \showISBNx    \undefined \def \showISBNx     #1{\unskip}     \fi
\ifx \showISBNxiii \undefined \def \showISBNxiii  #1{\unskip}     \fi
\ifx \showISSN     \undefined \def \showISSN      #1{\unskip}     \fi
\ifx \showLCCN     \undefined \def \showLCCN      #1{\unskip}     \fi
\ifx \shownote     \undefined \def \shownote      #1{#1}          \fi
\ifx \showarticletitle \undefined \def \showarticletitle #1{#1}   \fi
\ifx \showURL      \undefined \def \showURL       {\relax}        \fi
\providecommand\bibfield[2]{#2}
\providecommand\bibinfo[2]{#2}
\providecommand\natexlab[1]{#1}
\providecommand\showeprint[2][]{arXiv:#2}

\bibitem[Abas(2017)]%
        {abas_2017}
\bibfield{author}{\bibinfo{person}{Marcel Abas}.} \bibinfo{year}{2017}\natexlab{}.
\newblock \showarticletitle{Large Networks of Diameter Two Based on Cayley Graphs}. In \bibinfo{booktitle}{\emph{Computer Science On-line Conference}}. Springer, \bibinfo{pages}{225--233}.
\newblock


\bibitem[Ahn et~al\mbox{.}(2009)]%
        {ahn2009hyperx}
\bibfield{author}{\bibinfo{person}{Jung~Ho Ahn}, \bibinfo{person}{Nathan Binkert}, \bibinfo{person}{Al Davis}, \bibinfo{person}{Moray McLaren}, {and} \bibinfo{person}{Robert~S Schreiber}.} \bibinfo{year}{2009}\natexlab{}.
\newblock \showarticletitle{HyperX: topology, routing, and packaging of efficient large-scale networks}. In \bibinfo{booktitle}{\emph{Proceedings of the Conference on High Performance Computing Networking, Storage and Analysis}}. \bibinfo{pages}{1--11}.
\newblock


\bibitem[Awaji et~al\mbox{.}(2013)]%
        {awaji2013optical}
\bibfield{author}{\bibinfo{person}{Yoshinari Awaji}, \bibinfo{person}{Kunimasa Saitoh}, {and} \bibinfo{person}{Shoichiro Matsuo}.} \bibinfo{year}{2013}\natexlab{}.
\newblock \bibinfo{booktitle}{\emph{Optical Fiber Telecommunications VIB: Chapter 13. Transmission Systems Using Multicore Fibers}}.
\newblock \bibinfo{publisher}{Elsevier Inc. Chapters}.
\newblock


\bibitem[Bahn and Bagherzadeh(2008)]%
        {bahn2008generic}
\bibfield{author}{\bibinfo{person}{Jun~Ho Bahn} {and} \bibinfo{person}{Nader Bagherzadeh}.} \bibinfo{year}{2008}\natexlab{}.
\newblock \showarticletitle{A generic traffic model for on-chip interconnection networks}.
\newblock \bibinfo{journal}{\emph{Network on Chip Architectures}} (\bibinfo{year}{2008}), \bibinfo{pages}{22}.
\newblock


\bibitem[Bannai and Ito(1973)]%
        {Bannai1973OnFM}
\bibfield{author}{\bibinfo{person}{Eiichi Bannai} {and} \bibinfo{person}{Tatsuro Ito}.} \bibinfo{year}{1973}\natexlab{}.
\newblock \showarticletitle{On finite {M}oore graphs}.
\newblock \bibinfo{journal}{\emph{Journal of the Faculty of Science, the University of Tokyo. Sect. 1 A, Mathematics}}  \bibinfo{volume}{20} (\bibinfo{year}{1973}), \bibinfo{pages}{191--208}.
\newblock


\bibitem[Bermond et~al\mbox{.}(1982)]%
        {bermond82}
\bibfield{author}{\bibinfo{person}{J.C. Bermond}, \bibinfo{person}{C. Delorme}, {and} \bibinfo{person}{G. Farhi}.} \bibinfo{year}{1982}\natexlab{}.
\newblock \showarticletitle{{Large graphs with given degree and diameter III}}.
\newblock \bibinfo{journal}{\emph{Ann. of Discrete Math.}}  \bibinfo{volume}{13} (\bibinfo{year}{1982}), \bibinfo{pages}{23--32}.
\newblock


\bibitem[Besta and Hoefler(2014)]%
        {besta2014slim}
\bibfield{author}{\bibinfo{person}{Maciej Besta} {and} \bibinfo{person}{Torsten Hoefler}.} \bibinfo{year}{2014}\natexlab{}.
\newblock \showarticletitle{Slim {F}ly: A cost effective low-diameter network topology}. In \bibinfo{booktitle}{\emph{SC'14: proceedings of the international conference for high performance computing, networking, storage and analysis}}. IEEE, \bibinfo{pages}{348--359}.
\newblock


\bibitem[Besta and Hoefler(2015)]%
        {besta2015accelerating}
\bibfield{author}{\bibinfo{person}{Maciej Besta} {and} \bibinfo{person}{Torsten Hoefler}.} \bibinfo{year}{2015}\natexlab{}.
\newblock \showarticletitle{Accelerating irregular computations with hardware transactional memory and active messages}. In \bibinfo{booktitle}{\emph{Proceedings of the 24th International Symposium on High-Performance Parallel and Distributed Computing}}. \bibinfo{pages}{161--172}.
\newblock


\bibitem[Birrittella et~al\mbox{.}(2015)]%
        {birrittella2015intel}
\bibfield{author}{\bibinfo{person}{Mark~S Birrittella}, \bibinfo{person}{Mark Debbage}, \bibinfo{person}{Ram Huggahalli}, \bibinfo{person}{James Kunz}, \bibinfo{person}{Tom Lovett}, \bibinfo{person}{Todd Rimmer}, \bibinfo{person}{Keith~D Underwood}, {and} \bibinfo{person}{Robert~C Zak}.} \bibinfo{year}{2015}\natexlab{}.
\newblock \showarticletitle{Intel{\textregistered} omni-path architecture: Enabling scalable, high performance fabrics}. In \bibinfo{booktitle}{\emph{2015 IEEE 23rd Annual Symposium on High-Performance Interconnects}}. IEEE, \bibinfo{pages}{1--9}.
\newblock


\bibitem[Brahme et~al\mbox{.}(2013)]%
        {brahme2013symsig}
\bibfield{author}{\bibinfo{person}{Dhananjay Brahme}, \bibinfo{person}{Onkar Bhardwaj}, {and} \bibinfo{person}{Vipin Chaudhary}.} \bibinfo{year}{2013}\natexlab{}.
\newblock \showarticletitle{SymSig: A low latency interconnection topology for HPC clusters}. In \bibinfo{booktitle}{\emph{20th Annual International Conference on High Performance Computing}}. IEEE, \bibinfo{pages}{462--471}.
\newblock


\bibitem[Brown(1966)]%
        {brown_1966}
\bibfield{author}{\bibinfo{person}{W.~G. Brown}.} \bibinfo{year}{1966}\natexlab{}.
\newblock \showarticletitle{On Graphs that do not Contain a {T}homsen Graph}.
\newblock \bibinfo{journal}{\emph{Can. Math. Bull.}} \bibinfo{volume}{9}, \bibinfo{number}{3} (\bibinfo{year}{1966}), \bibinfo{pages}{281–285}.
\newblock
\urldef\tempurl%
\url{https://doi.org/10.4153/CMB-1966-036-2}
\showDOI{\tempurl}


\bibitem[Camarero et~al\mbox{.}(2016)]%
        {camarero2016projective}
\bibfield{author}{\bibinfo{person}{Crist{\'o}bal Camarero}, \bibinfo{person}{Carmen Mart{\'\i}nez}, \bibinfo{person}{Enrique Vallejo}, {and} \bibinfo{person}{Ram{\'o}n Beivide}.} \bibinfo{year}{2016}\natexlab{}.
\newblock \showarticletitle{Projective networks: Topologies for large parallel computer systems}.
\newblock \bibinfo{journal}{\emph{IEEE Transactions on Parallel and Distributed Systems}} \bibinfo{volume}{28}, \bibinfo{number}{7} (\bibinfo{year}{2016}), \bibinfo{pages}{2003--2016}.
\newblock


\bibitem[Choi(2022)]%
        {frontier}
\bibfield{author}{\bibinfo{person}{Charles~Q Choi}.} \bibinfo{year}{2022}\natexlab{}.
\newblock \bibinfo{title}{The Beating Heart of the World’s First Exascale Supercomputer}.
\newblock
\newblock
\newblock
\shownote{\url{https://spectrum.ieee.org/frontier-exascale-supercomputer}}.


\bibitem[Dalfó(2019)]%
        {Dalfo20191full}
\bibfield{author}{\bibinfo{person}{C. Dalfó}.} \bibinfo{year}{2019}\natexlab{}.
\newblock \showarticletitle{A survey on the missing {M}oore graph}.
\newblock \bibinfo{journal}{\emph{Linear Algebra Appl.}} (\bibinfo{year}{2019}).
\newblock


\bibitem[Damerell(1973)]%
        {Damerell1973OnMG}
\bibfield{author}{\bibinfo{person}{R.~M. Damerell}.} \bibinfo{year}{1973}\natexlab{}.
\newblock \showarticletitle{On {M}oore graphs}.
\newblock \bibinfo{journal}{\emph{Proc. Camb. Phil. Soc.}}  \bibinfo{volume}{74} (\bibinfo{year}{1973}), \bibinfo{pages}{227--236}.
\newblock


\bibitem[Dang et~al\mbox{.}(2018)]%
        {dang2018lightweight}
\bibfield{author}{\bibinfo{person}{Hoang-Vu Dang}, \bibinfo{person}{Roshan Dathathri}, \bibinfo{person}{Gurbinder Gill}, \bibinfo{person}{Alex Brooks}, \bibinfo{person}{Nikoli Dryden}, \bibinfo{person}{Andrew Lenharth}, \bibinfo{person}{Loc Hoang}, \bibinfo{person}{Keshav Pingali}, {and} \bibinfo{person}{Marc Snir}.} \bibinfo{year}{2018}\natexlab{}.
\newblock \showarticletitle{A lightweight communication runtime for distributed graph analytics}. In \bibinfo{booktitle}{\emph{2018 IEEE International Parallel and Distributed Processing Symposium (IPDPS)}}. IEEE, \bibinfo{pages}{980--989}.
\newblock


\bibitem[Daudlin et~al\mbox{.}(2021)]%
        {3d-multichip:Bergman}
\bibfield{author}{\bibinfo{person}{Stuart Daudlin}, \bibinfo{person}{Anthony Rizzo}, \bibinfo{person}{Nathan~C Abrams}, \bibinfo{person}{Sunwoo Lee}, \bibinfo{person}{Devesh Khilwani}, \bibinfo{person}{Vaishnavi Murthy}, \bibinfo{person}{James Robinson}, \bibinfo{person}{Terence Collier}, \bibinfo{person}{Alyosha Molnar}, {and} \bibinfo{person}{Keren Bergman}.} \bibinfo{year}{2021}\natexlab{}.
\newblock \showarticletitle{{3D-Integrated Multichip Module Transceiver for Terabit-Scale {DWDM} Interconnects}}. In \bibinfo{booktitle}{\emph{Optical Fiber Communications, OFC 2021}}.
\newblock


\bibitem[Dawkins et~al\mbox{.}(2024)]%
        {dawkins2024edge}
\bibfield{author}{\bibinfo{person}{Aleyah Dawkins}, \bibinfo{person}{Kelly Isham}, \bibinfo{person}{Ales Kubicek}, \bibinfo{person}{Kartik Lakhotia}, {and} \bibinfo{person}{Laura Monroe}.} \bibinfo{year}{2024}\natexlab{}.
\newblock \showarticletitle{Edge-Disjoint Spanning Trees on Star-Product Networks}.
\newblock \bibinfo{journal}{\emph{arXiv preprint arXiv:2403.12231}} (\bibinfo{year}{2024}).
\newblock


\bibitem[Dean and Barroso(2013)]%
        {Dean:tail-latency}
\bibfield{author}{\bibinfo{person}{Jeffrey Dean} {and} \bibinfo{person}{Luiz~André Barroso}.} \bibinfo{year}{2013}\natexlab{}.
\newblock \showarticletitle{{The Tail at Scale}}.
\newblock \bibinfo{journal}{\emph{Commun. ACM}}  \bibinfo{volume}{56} (\bibinfo{year}{2013}), \bibinfo{pages}{74--80}.
\newblock
\urldef\tempurl%
\url{http://cacm.acm.org/magazines/2013/2/160173-the-tail-at-scale/fulltext}
\showURL{%
\tempurl}


\bibitem[Dongarra(2020)]%
        {Fugaku:Dongarra}
\bibfield{author}{\bibinfo{person}{Jack Dongarra}.} \bibinfo{year}{2020}\natexlab{}.
\newblock \bibinfo{booktitle}{\emph{{Report on the {F}ujitsu {F}ugaku System}}}.
\newblock \bibinfo{type}{{T}echnical {R}eport} ICL-UT-20-06. \bibinfo{institution}{University of Tennessee, Knoxville}.
\newblock


\bibitem[Erd{\H o}s and R\'enyi(1962)]%
        {erdosrenyi1962}
\bibfield{author}{\bibinfo{person}{Paul Erd{\H o}s} {and} \bibinfo{person}{Alfred R\'enyi}.} \bibinfo{year}{1962}\natexlab{}.
\newblock \showarticletitle{On a problem in the theory of graphs}.
\newblock \bibinfo{journal}{\emph{Publ. Math. Inst. Hungar. Acad. Sci.}}  \bibinfo{volume}{7A} (\bibinfo{year}{1962}), \bibinfo{pages}{623--641}.
\newblock


\bibitem[Erdos and R{\'e}nyi(1963)]%
        {Erdos1963AsymmetricG}
\bibfield{author}{\bibinfo{person}{Paul~L. Erdos} {and} \bibinfo{person}{Alfr{\'e}d R{\'e}nyi}.} \bibinfo{year}{1963}\natexlab{}.
\newblock \showarticletitle{Asymmetric graphs}.
\newblock \bibinfo{journal}{\emph{Acta Mathematica Academiae Scientiarum Hungarica}}  \bibinfo{volume}{14} (\bibinfo{year}{1963}), \bibinfo{pages}{295--315}.
\newblock


\bibitem[Flajslik et~al\mbox{.}(2018)]%
        {flajslik2018megafly}
\bibfield{author}{\bibinfo{person}{Mario Flajslik}, \bibinfo{person}{Eric Borch}, {and} \bibinfo{person}{Mike~A Parker}.} \bibinfo{year}{2018}\natexlab{}.
\newblock \showarticletitle{Megafly: A topology for exascale systems}. In \bibinfo{booktitle}{\emph{High Performance Computing (ISC High Performance 2018)}}. Springer, \bibinfo{pages}{289--310}.
\newblock


\bibitem[Foley and Danskin(2017)]%
        {foley2017ultra}
\bibfield{author}{\bibinfo{person}{Denis Foley} {and} \bibinfo{person}{John Danskin}.} \bibinfo{year}{2017}\natexlab{}.
\newblock \showarticletitle{Ultra-performance Pascal GPU and NVLink interconnect}.
\newblock \bibinfo{journal}{\emph{IEEE Micro}} \bibinfo{volume}{37}, \bibinfo{number}{2} (\bibinfo{year}{2017}), \bibinfo{pages}{7--17}.
\newblock


\bibitem[Gold and Rader(1969)]%
        {gold1969digital}
\bibfield{author}{\bibinfo{person}{Bernard Gold} {and} \bibinfo{person}{Charles~M Rader}.} \bibinfo{year}{1969}\natexlab{}.
\newblock \showarticletitle{Digital processing of signals}.
\newblock \bibinfo{journal}{\emph{Digital processing of signals}} (\bibinfo{year}{1969}).
\newblock


\bibitem[Hammond et~al\mbox{.}(2015)]%
        {hammond2015ember}
\bibfield{author}{\bibinfo{person}{Simon~David Hammond}, \bibinfo{person}{Karl~Scott Hemmert}, \bibinfo{person}{Michael~J Levenhagen}, \bibinfo{person}{Arun~F Rodrigues}, {and} \bibinfo{person}{Gwendolyn~Renae Voskuilen}.} \bibinfo{year}{2015}\natexlab{}.
\newblock \bibinfo{booktitle}{\emph{Ember: Reference Communication Patterns for Exascale.}}
\newblock \bibinfo{type}{{T}echnical {R}eport}. \bibinfo{institution}{Sandia National Lab.(SNL-NM), Albuquerque, NM (United States)}.
\newblock


\bibitem[Hemmert(2018)]%
        {hemmert2018merlin}
\bibfield{author}{\bibinfo{person}{Karl~Scott Hemmert}.} \bibinfo{year}{2018}\natexlab{}.
\newblock \bibinfo{booktitle}{\emph{Merlin Element Library Deep Dive.}}
\newblock \bibinfo{type}{{T}echnical {R}eport}. \bibinfo{institution}{Sandia National Lab.(SNL-NM), Albuquerque, NM (United States)}.
\newblock


\bibitem[Hoffman and Singleton(1960)]%
        {hoffmansingleton1960}
\bibfield{author}{\bibinfo{person}{A.~J. Hoffman} {and} \bibinfo{person}{R.~R. Singleton}.} \bibinfo{year}{1960}\natexlab{}.
\newblock \showarticletitle{On {M}oore Graphs with Diameters 2 and 3}.
\newblock \bibinfo{journal}{\emph{IBM Journal of Research and Development}} \bibinfo{volume}{4}, \bibinfo{number}{5} (\bibinfo{year}{1960}), \bibinfo{pages}{497--504}.
\newblock
\urldef\tempurl%
\url{https://doi.org/10.1147/rd.45.0497}
\showDOI{\tempurl}


\bibitem[Hoisie et~al\mbox{.}(2007)]%
        {hoisie2007performance}
\bibfield{author}{\bibinfo{person}{Adolfy Hoisie}, \bibinfo{person}{Olaf Lubeck}, {and} \bibinfo{person}{Harvey Wasserman}.} \bibinfo{year}{2007}\natexlab{}.
\newblock \showarticletitle{Performance analysis of wavefront algorithms on very-large scale distributed systems}. In \bibinfo{booktitle}{\emph{Workshop on wide area networks and high performance computing}}. Springer, \bibinfo{pages}{171--187}.
\newblock


\bibitem[Jiang et~al\mbox{.}(2013)]%
        {jiang2013detailed}
\bibfield{author}{\bibinfo{person}{Nan Jiang}, \bibinfo{person}{Daniel~U Becker}, \bibinfo{person}{George Michelogiannakis}, \bibinfo{person}{James Balfour}, \bibinfo{person}{Brian Towles}, \bibinfo{person}{David~E Shaw}, \bibinfo{person}{John Kim}, {and} \bibinfo{person}{William~J Dally}.} \bibinfo{year}{2013}\natexlab{}.
\newblock \showarticletitle{A detailed and flexible cycle-accurate network-on-chip simulator}. In \bibinfo{booktitle}{\emph{2013 IEEE international symposium on performance analysis of systems and software (ISPASS)}}. IEEE.
\newblock


\bibitem[Kachris and Tomkos(2012)]%
        {kachris2012survey}
\bibfield{author}{\bibinfo{person}{Christoforos Kachris} {and} \bibinfo{person}{Ioannis Tomkos}.} \bibinfo{year}{2012}\natexlab{}.
\newblock \showarticletitle{A survey on optical interconnects for data centers}.
\newblock \bibinfo{journal}{\emph{IEEE Communications Surveys \& Tutorials}} \bibinfo{volume}{14}, \bibinfo{number}{4} (\bibinfo{year}{2012}), \bibinfo{pages}{1021--1036}.
\newblock


\bibitem[Karypis and Kumar(2009)]%
        {METIS}
\bibfield{author}{\bibinfo{person}{George Karypis} {and} \bibinfo{person}{Vipin Kumar}.} \bibinfo{year}{2009}\natexlab{}.
\newblock \bibinfo{title}{{MeTis: Unstructured Graph Partitioning and Sparse Matrix Ordering System, Version 4.0}}.
\newblock \bibinfo{howpublished}{\url{https://github.com/KarypisLab/METIS}}.
\newblock


\bibitem[Kathareios et~al\mbox{.}(2015)]%
        {diameter-2-topos}
\bibfield{author}{\bibinfo{person}{G. Kathareios}, \bibinfo{person}{C. Minkenberg}, \bibinfo{person}{B. Prisacari}, \bibinfo{person}{G. Rodriguez}, {and} \bibinfo{person}{Torsten Hoefler}.} \bibinfo{year}{2015}\natexlab{}.
\newblock \showarticletitle{{Cost-Effective Diameter-Two Topologies: Analysis and Evaluation}}. In \bibinfo{booktitle}{\emph{Proceedings of the International Conference for High Performance Computing, Networking, Storage and Analysis}} (Austin, TX, USA). \bibinfo{publisher}{ACM}.
\newblock
\showISBNx{978-1-4503-3723-6}


\bibitem[Keeler({[n.\,d.]})]%
        {darpa-eri-2019}
\bibfield{author}{\bibinfo{person}{Gordon Keeler}.} \bibinfo{year}{[n.\,d.]}\natexlab{}.
\newblock \bibinfo{title}{{ERI Programs Panel - Phase II Overview}}.
\newblock
\newblock
\newblock
\shownote{DARPA ERI Summit 2019}.


\bibitem[Kim et~al\mbox{.}(2007)]%
        {kim2007flattened}
\bibfield{author}{\bibinfo{person}{John Kim}, \bibinfo{person}{William~J Dally}, {and} \bibinfo{person}{Dennis Abts}.} \bibinfo{year}{2007}\natexlab{}.
\newblock \showarticletitle{Flattened butterfly: a cost-efficient topology for high-radix networks}. In \bibinfo{booktitle}{\emph{Proceedings of the 34th annual international symposium on Computer architecture}}. \bibinfo{pages}{126--137}.
\newblock


\bibitem[Kim et~al\mbox{.}(2008)]%
        {dally08}
\bibfield{author}{\bibinfo{person}{John Kim}, \bibinfo{person}{Wiliam~J. Dally}, \bibinfo{person}{Steve Scott}, {and} \bibinfo{person}{Dennis Abts}.} \bibinfo{year}{2008}\natexlab{}.
\newblock \showarticletitle{Technology-Driven, Highly-Scalable {D}ragonfly Topology}. In \bibinfo{booktitle}{\emph{Proceedings of the 35th ISCA}}. \bibinfo{publisher}{IEEE Computer Society}, \bibinfo{address}{Washington, DC, USA}.
\newblock


\bibitem[Lakhotia et~al\mbox{.}(2022)]%
        {polarfly_sc22}
\bibfield{author}{\bibinfo{person}{Kartik Lakhotia}, \bibinfo{person}{Maciej Besta}, \bibinfo{person}{Laura Monroe}, \bibinfo{person}{Kelly Isham}, \bibinfo{person}{Patrick Iff}, \bibinfo{person}{Torsten Hoefler}, {and} \bibinfo{person}{Fabrizio Petrini}.} \bibinfo{year}{2022}\natexlab{}.
\newblock \showarticletitle{PolarFly: A Cost-Effective and Flexible Low-Diameter Topology}. In \bibinfo{booktitle}{\emph{Proceedings of the International Conference on High Performance Computing, Networking, Storage and Analysis}}. IEEE, \bibinfo{pages}{1--15}.
\newblock


\bibitem[Lakhotia et~al\mbox{.}(2023)]%
        {lakhotia2023network}
\bibfield{author}{\bibinfo{person}{Kartik Lakhotia}, \bibinfo{person}{Kelly Isham}, \bibinfo{person}{Laura Monroe}, \bibinfo{person}{Maciej Besta}, \bibinfo{person}{Torsten Hoefler}, {and} \bibinfo{person}{Fabrizio Petrini}.} \bibinfo{year}{2023}\natexlab{}.
\newblock \showarticletitle{In-network Allreduce with Multiple Spanning Trees on PolarFly}. In \bibinfo{booktitle}{\emph{Proceedings of the 35th ACM Symposium on Parallelism in Algorithms and Architectures}}. \bibinfo{pages}{165--176}.
\newblock


\bibitem[Lakhotia et~al\mbox{.}(2021)]%
        {lakhotia2021accelerating}
\bibfield{author}{\bibinfo{person}{Kartik Lakhotia}, \bibinfo{person}{Fabrizio Petrini}, \bibinfo{person}{Rajgopal Kannan}, {and} \bibinfo{person}{Viktor Prasanna}.} \bibinfo{year}{2021}\natexlab{}.
\newblock \showarticletitle{Accelerating {A}llreduce with in-network reduction on {I}ntel {PIUMA}}.
\newblock \bibinfo{journal}{\emph{IEEE Micro}} \bibinfo{volume}{42}, \bibinfo{number}{2} (\bibinfo{year}{2021}), \bibinfo{pages}{44--52}.
\newblock


\bibitem[Lei et~al\mbox{.}(2020)]%
        {lei2020bundlefly}
\bibfield{author}{\bibinfo{person}{Fei Lei}, \bibinfo{person}{Dezun Dong}, \bibinfo{person}{Xiangke Liao}, {and} \bibinfo{person}{Jos{\'e} Duato}.} \bibinfo{year}{2020}\natexlab{}.
\newblock \showarticletitle{Bundlefly: A low-diameter topology for multicore fiber}. In \bibinfo{booktitle}{\emph{Proceedings of the 34th ACM International Conference on Supercomputing}}.
\newblock


\bibitem[Lei et~al\mbox{.}(2016)]%
        {lei2016galaxyfly}
\bibfield{author}{\bibinfo{person}{Fei Lei}, \bibinfo{person}{Dezun Dong}, \bibinfo{person}{Xiangke Liao}, \bibinfo{person}{Xing Su}, {and} \bibinfo{person}{Cunlu Li}.} \bibinfo{year}{2016}\natexlab{}.
\newblock \showarticletitle{Galaxyfly: A Novel Family of Flexible-Radix Low-Diameter Topologies for Large-Scales Interconnection Networks}. In \bibinfo{booktitle}{\emph{Proceedings of the 2016 International Conference on Supercomputing}} (Istanbul, Turkey) \emph{(\bibinfo{series}{ICS '16})}. \bibinfo{publisher}{Association for Computing Machinery}, \bibinfo{address}{New York, NY, USA}, Article \bibinfo{articleno}{24}, \bibinfo{numpages}{12}~pages.
\newblock
\showISBNx{9781450343619}
\urldef\tempurl%
\url{https://doi.org/10.1145/2925426.2926275}
\showDOI{\tempurl}


\bibitem[Leiserson(1985a)]%
        {leiserson1985fat}
\bibfield{author}{\bibinfo{person}{Charles~E Leiserson}.} \bibinfo{year}{1985}\natexlab{a}.
\newblock \showarticletitle{Fat-trees: Universal networks for hardware-efficient supercomputing}.
\newblock \bibinfo{journal}{\emph{IEEE transactions on Computers}} \bibinfo{volume}{100}, \bibinfo{number}{10} (\bibinfo{year}{1985}), \bibinfo{pages}{892--901}.
\newblock


\bibitem[Leiserson(1985b)]%
        {Leiserson:1985:FUN:4492.4495}
\bibfield{author}{\bibinfo{person}{Charles~E. Leiserson}.} \bibinfo{year}{1985}\natexlab{b}.
\newblock \showarticletitle{{Fat-trees: universal networks for hardware-efficient supercomputing}}.
\newblock \bibinfo{journal}{\emph{IEEE Trans. Comput.}} \bibinfo{volume}{34}, \bibinfo{number}{10} (\bibinfo{date}{Oct.} \bibinfo{year}{1985}), \bibinfo{pages}{892--901}.
\newblock
\showISSN{0018-9340}


\bibitem[Li et~al\mbox{.}(2004)]%
        {li2004graph}
\bibfield{author}{\bibinfo{person}{Dongsheng Li}, \bibinfo{person}{Xicheng Lu}, {and} \bibinfo{person}{Jinshu Su}.} \bibinfo{year}{2004}\natexlab{}.
\newblock \showarticletitle{Graph-theoretic analysis of Kautz topology and DHT schemes}. In \bibinfo{booktitle}{\emph{IFIP International Conference on Network and Parallel Computing}}. Springer, \bibinfo{pages}{308--315}.
\newblock


\bibitem[{Lightmatter}({[n.\,d.]})]%
        {lightmatter}
\bibfield{author}{\bibinfo{person}{{Lightmatter}}.} \bibinfo{year}{[n.\,d.]}\natexlab{}.
\newblock
\newblock
\newblock
\shownote{\url{https://lightmatter.co/}}.


\bibitem[Loz et~al\mbox{.}(2010)]%
        {comb_wiki_degdiam_general}
\bibfield{author}{\bibinfo{person}{E Loz}, \bibinfo{person}{H. P\'erez-Ros\'es}, {and} \bibinfo{person}{G.Pineda-Villavicencio}.} \bibinfo{year}{2010}\natexlab{}.
\newblock \bibinfo{title}{The Degree-Diameter Problem for General Graphs}.
\newblock
\newblock
\newblock
\shownote{\url{http://www.combinatoricswiki.org/wiki/The_Degree_Diameter_Problem_for_General_Graphs}}.


\bibitem[McEliece(1987)]%
        {mceliece_1987}
\bibfield{author}{\bibinfo{person}{Robert~J. McEliece}.} \bibinfo{year}{1987}\natexlab{}.
\newblock \bibinfo{booktitle}{\emph{Finite Fields for Computer Scientists and Engineers}}.
\newblock \bibinfo{publisher}{Springer}, \bibinfo{address}{Boston, MA}.
\newblock


\bibitem[McKay et~al\mbox{.}(1998)]%
        {mckay98}
\bibfield{author}{\bibinfo{person}{Brendan~D. McKay}, \bibinfo{person}{Mirka Miller}, {and} \bibinfo{person}{Jozef \v{S}ir\'{a}n}.} \bibinfo{year}{1998}\natexlab{}.
\newblock \showarticletitle{A note on Large Graphs of Diameter Two and Given Maximum Degree}.
\newblock \bibinfo{journal}{\emph{Journal of Combinatorial Theory, Series B}} (\bibinfo{year}{1998}).
\newblock


\bibitem[Meuer et~al\mbox{.}(2023)]%
        {top500}
\bibfield{author}{\bibinfo{person}{Hans Meuer}, \bibinfo{person}{Erich Strohmaier}, \bibinfo{person}{Jack Dongarra}, \bibinfo{person}{Horst Simon}, {and} \bibinfo{person}{Martin Meuer}.} \bibinfo{year}{2023}\natexlab{}.
\newblock \bibinfo{title}{{Top 500: The List}}.
\newblock \bibinfo{howpublished}{\url{https://top500.org/lists/top500/}}.
\newblock


\bibitem[Paley(1933)]%
        {Paley1933OnOM}
\bibfield{author}{\bibinfo{person}{R~E Paley}.} \bibinfo{year}{1933}\natexlab{}.
\newblock \showarticletitle{On Orthogonal Matrices}.
\newblock \bibinfo{journal}{\emph{Journal of Mathematics and Physics}}  \bibinfo{volume}{12} (\bibinfo{year}{1933}), \bibinfo{pages}{311--320}.
\newblock


\bibitem[Pfister(2001)]%
        {pfister2001introduction}
\bibfield{author}{\bibinfo{person}{Gregory~F Pfister}.} \bibinfo{year}{2001}\natexlab{}.
\newblock \showarticletitle{An introduction to the infiniband architecture}.
\newblock \bibinfo{journal}{\emph{High performance mass storage and parallel I/O}} \bibinfo{volume}{42}, \bibinfo{number}{617-632} (\bibinfo{year}{2001}), \bibinfo{pages}{10}.
\newblock


\bibitem[Prikopa et~al\mbox{.}(2016)]%
        {prikopa2016parallel}
\bibfield{author}{\bibinfo{person}{Karl~E Prikopa}, \bibinfo{person}{Wilfried~N Gansterer}, {and} \bibinfo{person}{Elias Wimmer}.} \bibinfo{year}{2016}\natexlab{}.
\newblock \showarticletitle{Parallel iterative refinement linear least squares solvers based on all-reduce operations}.
\newblock \bibinfo{journal}{\emph{Parallel Comput.}}  \bibinfo{volume}{57} (\bibinfo{year}{2016}), \bibinfo{pages}{167--184}.
\newblock


\bibitem[Rabenseifner(2004)]%
        {rabenseifner2004optimization}
\bibfield{author}{\bibinfo{person}{Rolf Rabenseifner}.} \bibinfo{year}{2004}\natexlab{}.
\newblock \showarticletitle{Optimization of collective reduction operations}. In \bibinfo{booktitle}{\emph{International Conference on Computational Science}}, Vol.~\bibinfo{volume}{3036}. \bibinfo{pages}{1--9}.
\newblock


\bibitem[Rodrigues et~al\mbox{.}(2011)]%
        {rodrigues2011structural}
\bibfield{author}{\bibinfo{person}{Arun~F Rodrigues}, \bibinfo{person}{K~Scott Hemmert}, \bibinfo{person}{Brian~W Barrett}, \bibinfo{person}{Chad Kersey}, \bibinfo{person}{Ron Oldfield}, \bibinfo{person}{Marlo Weston}, \bibinfo{person}{Rolf Risen}, \bibinfo{person}{Jeanine Cook}, \bibinfo{person}{Paul Rosenfeld}, \bibinfo{person}{Elliot Cooper-Balis}, {et~al\mbox{.}}} \bibinfo{year}{2011}\natexlab{}.
\newblock \showarticletitle{The structural simulation toolkit}.
\newblock \bibinfo{journal}{\emph{ACM SIGMETRICS Performance Evaluation Review}} \bibinfo{volume}{38}, \bibinfo{number}{4} (\bibinfo{year}{2011}), \bibinfo{pages}{37--42}.
\newblock


\bibitem[Sarbazi-Azad et~al\mbox{.}(2001)]%
        {sarbazi2001communication}
\bibfield{author}{\bibinfo{person}{Hamid Sarbazi-Azad}, \bibinfo{person}{Mohamed Ould-Khaoua}, {and} \bibinfo{person}{Lewis~M. Mackenzie}.} \bibinfo{year}{2001}\natexlab{}.
\newblock \showarticletitle{Communication delay in hypercubes in the presence of bit-reversal traffic}.
\newblock \bibinfo{journal}{\emph{Parallel Comput.}} \bibinfo{volume}{27}, \bibinfo{number}{13} (\bibinfo{year}{2001}), \bibinfo{pages}{1801--1816}.
\newblock


\bibitem[Sergeev and Del~Balso(2018)]%
        {sergeev2018horovod}
\bibfield{author}{\bibinfo{person}{Alexander Sergeev} {and} \bibinfo{person}{Mike Del~Balso}.} \bibinfo{year}{2018}\natexlab{}.
\newblock \showarticletitle{Horovod: fast and easy distributed deep learning in TensorFlow}.
\newblock \bibinfo{journal}{\emph{arXiv preprint arXiv:1802.05799}} (\bibinfo{year}{2018}).
\newblock


\bibitem[Shpiner et~al\mbox{.}(2017)]%
        {shpiner2017dragonfly+}
\bibfield{author}{\bibinfo{person}{Alexander Shpiner}, \bibinfo{person}{Zachy Haramaty}, \bibinfo{person}{Saar Eliad}, \bibinfo{person}{Vladimir Zdornov}, \bibinfo{person}{Barak Gafni}, {and} \bibinfo{person}{Eitan Zahavi}.} \bibinfo{year}{2017}\natexlab{}.
\newblock \showarticletitle{Dragonfly+: Low cost topology for scaling datacenters}. In \bibinfo{booktitle}{\emph{2017 IEEE 3rd International Workshop on High-Performance Interconnection Networks in the Exascale and Big-Data Era (HiPINEB)}}. IEEE, \bibinfo{pages}{1--8}.
\newblock


\bibitem[Singla et~al\mbox{.}(2012)]%
        {Singla:2012:JND:2228298.2228322}
\bibfield{author}{\bibinfo{person}{Ankit Singla}, \bibinfo{person}{Chi-Yao Hong}, \bibinfo{person}{Lucian Popa}, {and} \bibinfo{person}{P.~Brighten Godfrey}.} \bibinfo{year}{2012}\natexlab{}.
\newblock \showarticletitle{{Jellyfish: networking data centers randomly}}. In \bibinfo{booktitle}{\emph{{Proceedings of the 9th USENIX NSDI}}} (San Jose, CA). \bibinfo{publisher}{USENIX Association}, \bibinfo{address}{Berkeley, CA, USA}, \bibinfo{pages}{17--17}.
\newblock


\bibitem[Stewart and Gingold(2006)]%
        {stewart2006new}
\bibfield{author}{\bibinfo{person}{Lawrence~C Stewart} {and} \bibinfo{person}{David Gingold}.} \bibinfo{year}{2006}\natexlab{}.
\newblock \showarticletitle{A new generation of cluster interconnect}.
\newblock \bibinfo{journal}{\emph{White Paper, SiCortex Inc}} (\bibinfo{year}{2006}).
\newblock


\bibitem[Torudbakken and Krishnamoorthy(2013)]%
        {torudbakken201350tbps}
\bibfield{author}{\bibinfo{person}{Ola Torudbakken} {and} \bibinfo{person}{Ashok~V Krishnamoorthy}.} \bibinfo{year}{2013}\natexlab{}.
\newblock \showarticletitle{A 50Tbps optically-cabled InfiniBand datacenter switch}. In \bibinfo{booktitle}{\emph{Optical Fiber Communication Conference}}. Optica Publishing Group, \bibinfo{pages}{OTu3H--1}.
\newblock


\bibitem[Wade et~al\mbox{.}(2020)]%
        {wade2020teraphy}
\bibfield{author}{\bibinfo{person}{Mark Wade}, \bibinfo{person}{Erik Anderson}, \bibinfo{person}{Shaha Ardalan}, \bibinfo{person}{Pavan Bhargava}, \bibinfo{person}{Sidney Buchbinder}, \bibinfo{person}{Michael Davenport}, \bibinfo{person}{John Fini}, \bibinfo{person}{Haiwei Lu}, \bibinfo{person}{Chen Li}, {and} \bibinfo{person}{Roy Meade}.} \bibinfo{year}{2020}\natexlab{}.
\newblock \showarticletitle{{TeraPHY: a Chiplet Technology for Low-Power, High-Bandwidth In-Package optical I/O}}.
\newblock \bibinfo{journal}{\emph{IEEE Micro}} \bibinfo{volume}{40}, \bibinfo{number}{2} (\bibinfo{year}{2020}), \bibinfo{pages}{63--71}.
\newblock


\bibitem[Wilber(1989)]%
        {wilber1989lower}
\bibfield{author}{\bibinfo{person}{Robert Wilber}.} \bibinfo{year}{1989}\natexlab{}.
\newblock \showarticletitle{Lower bounds for accessing binary search trees with rotations}.
\newblock \bibinfo{journal}{\emph{SIAM journal on Computing}} \bibinfo{volume}{18}, \bibinfo{number}{1} (\bibinfo{year}{1989}), \bibinfo{pages}{56--67}.
\newblock


\bibitem[Young et~al\mbox{.}(2022)]%
        {aksoy2021spectralfly}
\bibfield{author}{\bibinfo{person}{Stephen Young}, \bibinfo{person}{Sinan Aksoy}, \bibinfo{person}{Jesun Firoz}, \bibinfo{person}{Roberto Gioiosa}, \bibinfo{person}{Tobias Hagge}, \bibinfo{person}{Mark Kempton}, \bibinfo{person}{Juan Escobedo}, {and} \bibinfo{person}{Mark Raugas}.} \bibinfo{year}{2022}\natexlab{}.
\newblock \showarticletitle{SpectralFly: Ramanujan Graphs as Flexible and Efficient Interconnection Networks}. In \bibinfo{booktitle}{\emph{2022 IEEE International Parallel and Distributed Processing Symposium (IPDPS)}}. \bibinfo{pages}{1040--1050}.
\newblock


\end{thebibliography}
\end{document}